\documentclass[twoside]{article}
\usepackage{fancyhdr}
\usepackage[preprint]{aistats2027}
\usepackage[utf8]{inputenc}
\usepackage[T1]{fontenc}
\usepackage[hyphens]{url}
\usepackage{graphicx}
\usepackage[authoryear,round]{natbib}
\usepackage{booktabs}
\usepackage{float}
\usepackage{amsmath,amssymb,amsthm,mathtools}
\usepackage{xcolor}
\usepackage{rotating}
\usepackage[
  colorlinks=false,
  citecolor=green,
  linkcolor=red,
  urlcolor=cyan
]{hyperref}
\usepackage[nameinlink,noabbrev]{cleveref}
\crefname{appendix}{appendix}{appendices}
\Crefname{appendix}{Appendix}{Appendices}

\graphicspath{{./}{../}{../results_block3/}}
\definecolor{gray}{gray}{0.4}
\definecolor{orange}{RGB}{220,120,0}

\newtheorem{proposition}{Proposition}
\newtheorem{theorem}{Theorem}
\newtheorem{corollary}{Corollary}

\newtheorem{remark}{Remark}
\theoremstyle{definition}

\begin{document}

\runningtitle{Minimum Specification Perturbation}
\runningauthor{Dang, Pham, Nguyen}

\twocolumn[
\aistatstitle{Minimum Specification Perturbation: Robustness as Distance-to-Falsification in Causal Inference}

\aistatsauthor{Hoang Dang \And Luan Pham \And Minh Nguyen}

\aistatsaddress{
Independent Researcher \\
\scriptsize\texttt{\href{mailto:hoangdang112023@gmail.com}{hoangdang112023@gmail.com}}
\And
University of New South Wales \\
Australia \\
\scriptsize\texttt{\href{mailto:luan.pham@unsw.edu.au}{luan.pham@unsw.edu.au}}
\And
Florida Atlantic University \\
USA \\
\scriptsize\texttt{\href{mailto:minhnguyen@fau.edu}{minhnguyen@fau.edu}}
}
]

\begin{abstract}
How many analyst decisions, such as covariates or estimators, must change before the data stop supporting a claim? We introduce \emph{Minimum Specification Perturbation} (MSP): the fewest declared decisions to change in the reported analysis, holding the data fixed. The change must continue until the confidence interval contains zero or significantly reverses the claim. We count these two events separately, because an interval can reverse without touching zero. It matters because a claim can sit one changed decision from losing support even when most analyses agree with it. With pointwise intervals, the directional MSP provides a lower confidence bound for its population value. A simultaneous band bounds that same value sharply from both sides.
MSP is measured from the analysis and names the nearest decision changes that withdraw support, complementing other measures. In our reanalysis of a published LaLonde workflow, changing only the estimator to matching yields an interval containing zero. The nearest witness tells the researcher where to look next. We provide a workflow that turns that clue into a focused follow-up. For a matching witness there, the workflow traces its loss of support to tie-handling, not to the estimator itself.
\end{abstract}

\section{Introduction}
\label{sec:introduction}

Empirical causal claims depend on many defensible analyst decisions: estimators, covariates, functional forms, overlap rules, and sample restrictions. Each is a fork a different analyst could reasonably have taken~\citep{gelman2013garden,simmons2011false}. Multiverse analysis, specification-curve analysis, and many-analysts studies respond by enumerating such specifications and reporting how estimates are distributed across them~\citep{steegen2016increasing,simonsohn2020specification,silberzahn2018many}. What this distribution does not answer is how many decisions must change before the data stop supporting the claim. ``Most specifications are negative'' describes a cloud's location and spread, not the baseline's distance to its nearest loss-of-support neighbor. The gap is not definitional: two reports can share identical dispersion while only one baseline sits one decision from failing. Agreement across specifications therefore risks over-crediting that baseline~\citep{gelman2013garden,simmons2011false}.

We introduce \textbf{Minimum Specification Perturbation (MSP)}: the fewest decisions
that must change before a claim loses support, holding data fixed. A claim
loses support when an interval contains zero or significantly supports the
opposite direction.
We define $MSP_{null}$ (null contact) and $MSP_{rev}$ (significant reversal). Tracking $MSP_{null}$ alone is insufficient: finding no null contact can mean either universal support or an overshoot into significant reversal. We therefore report the pair $(MSP_{null}, MSP_{rev})$ and its directional minimum $MSP_{dir} = \min\{MSP_{null}, MSP_{rev}\}$. Both are exactly computable on any declared finite space, and, given a specification curve, they need no new model fits. Each also names a decision a referee can contest.

MSP is not a substitute for three adjacent literatures, but it perturbs a different object than each. Hidden-bias sensitivity analysis perturbs an unobserved confounder~\citep{rosenbaum2002observational,vanderweele2017sensitivity,oster2019unobservable,cinelli2020making}, while the Fragility Index and approximate maximum influence perturbation (AMIP) perturb the observations~\citep{walsh2014statistical,broderick2020automatic}. MSP, by contrast, holds the data fixed and measures distance in the declared analysis~\citep{goodfellow2014explaining,karimi2021algorithmic}.
 These measures need not move together. In a randomized job-training experiment, zeroing the earnings of the three highest-earning treated workers removes the baseline's support. Yet no declared change of analysis does, while a simulation shows the reverse (\Cref{tab:fragility-vs-msp,tab:fi-msp-flip-appendix}).

\paragraph{Toy example.} Consider a positive claim built from two ordinary
decisions: covariates (basic or full) and estimator (OLS or IPW). The
baseline uses basic covariates and OLS. All four combinations target the
same effect (\Cref{tab:toy-example}):
\begin{table}[H]
\centering
\setlength{\tabcolsep}{3pt}
\begin{tabular}{@{}lccc@{}}
\toprule
Specification & covariates & estimator & 95\% CI \\
\midrule
Baseline & basic & OLS & \([100,600]\) \\
Change covariates & full & OLS & \([-40,200]\) \\
Change estimator & basic & IPW & \([80,500]\) \\
Change both & full & IPW & \([-300,-30]\) \\
\bottomrule
\end{tabular}
\caption{Toy example: four analyses from two ordinary decisions. One
changed decision, the covariate set, already gives an interval containing
zero, so \(\mathrm{MSP}_{\mathrm{null}}=1\). A significant reversal takes
both, so \(\mathrm{MSP}_{\mathrm{rev}}=2\). Intervals are hypothetical.}
\label{tab:toy-example}
\end{table}
The baseline and the estimator-only change are entirely positive.
Changing covariates alone pulls the interval to \([-40,200]\), which contains
zero, so one decision suffices: \(\mathrm{MSP}_{\mathrm{null}}=1\). Changing both decisions
goes further, flipping the interval to \([-300,-30]\), entirely negative:
that takes two decisions, so \(\mathrm{MSP}_{\mathrm{rev}}=2\). Taking the
smaller of the two, \(\mathrm{MSP}_{\mathrm{dir}}=1\), gives the distance to
losing support.
 
% Our contributions follow this argument.
% \begin{enumerate}
%     \item We define directional MSP on arbitrary finite, feasible specification spaces. The reporting pair resolves the overshoot ambiguity of null-contact distance. Its minimum is the one-number distance to loss of directional support.
%     \item We give the pair a statistical interpretation. A pointwise lower-limit guarantee makes the directional distance conservative for its population target. A simultaneous confidence band, by contrast, yields a sharp envelope for that target. The same distance tests a declared robustness radius as an intersection--union test, so it needs no multiplicity correction. We also show that, under stated conditions, adding analyses can never raise the reported distance. It can lower the distance only if an added analysis nearer the baseline loses support. A distance of one, therefore, cannot move.
%     \item We evaluate MSP in two causal applications and a controlled simulation. We also run an exact stress test of all product-form sub-declarations of an enriched LaLonde space. Across these, MSP exposes sensitivity that selected distributional summaries do not encode. The stress test identifies when omitted options hide nearest witnesses and which analyst decisions are load-bearing. In the weak-effect regime of the simulation, the overshoot state is common.
% \end{enumerate}

\paragraph{Contributions.}
\begin{enumerate}

\item \textbf{Metric formulation.}
MSP is the fewest analyst decisions that must change, data fixed, before a claim loses support. We report it as a pair, so a flip without crossing zero is not missed. The smaller distance, plus its decisions, is the key number.

\item \textbf{Declaration geometry.}
A fuller valid declaration can only reduce the pointwise distance, never
increase it. The reported value is therefore the most favorable distance among its valid
enlargements. Under \(r\)-completeness, we provide sharp lower and upper bounds. Distance one remains exact under every such enlargement, without any
completeness assumption.

\item \textbf{Statistical inference.}
Pointwise coverage makes sample directional MSP a \((1-\alpha)\) lower confidence bound for its population target. The corresponding robustness-radius test needs no multiplicity adjustment for level control. Simultaneous bands can bound the distance from both sides, although the upper endpoint may remain infinite.

\item \textbf{Empirical divergence.}
LaLonde and Card--Krueger show which decisions are pivotal, and which thinner checks would have missed them. Simulations show a claim can look solid yet sit one decision from failing far more often than expected. MSP and the Fragility Index, which perturbs data, not analysis, can disagree sharply in either direction, confirming different risks.

\item \textbf{Witness-guided follow-up.}
We route nearest witnesses to follow-ups determined by their declared axis types. On LaLonde, the workflow traces a matching witness to tie handling and concentrated control weights. It also identifies pre-treatment earnings as minimal failing additions within the nested covariate axis.

\end{enumerate}

\section{Related Work}
\label{sec:related}

\textbf{Sensitivity analysis.} Hidden-bias methods ask how strong unmeasured confounding must be to overturn a conclusion~\citep{rosenbaum2002observational,imbens2003sensitivity,vanderweele2017sensitivity,oster2019unobservable,cinelli2020making,rambachan2023more}. Inference on breakdown frontiers studies a related question when the sensitivity parameter is continuous~\citep{masten2020inference}. MSP is complementary: it perturbs documented analytic decisions, not a model of unobserved bias or a continuous confounding parameter. It should be read alongside, not instead of, these sensitivity analyses. Our supplementary LaLonde comparison makes this explicit. Holding the baseline regression fixed leaves E-values, Oster's \(\delta\), and Cinelli--Hazlett robustness values unchanged across two specification landscapes. MSP changes with the landscape and its nearest witness (\Cref{tab:hidden-bias-joint}).

\textbf{Specification enumeration.} Multiverse and specification-curve analyses expose analytic flexibility by reporting estimates over defensible specifications~\citep{steegen2016increasing,simonsohn2020specification,silberzahn2018many}. This literature motivates reporting many analyses jointly, but MSP is not a pooled estimate or a meta-analytic average. It retains the declared graph of analyst decisions and asks for the nearest loss-of-support witness. Dispersion, configuration counts, and related instability summaries describe the landscape of estimates. MSP uses that landscape to measure distance from a baseline to loss of support.

\textbf{Minimum perturbations.} The approximate maximum influence perturbation (AMIP) framework uses influence-function approximations to identify a small subset or fraction of observations. Removing that subset can alter a fixed analysis~\citep{broderick2020automatic}. The Fragility Index and related influence diagnostics likewise perturb observations or outcomes~\citep{hampel1971general,walsh2014statistical}. Adversarial robustness and counterfactual explanations, meanwhile, seek small input changes that alter a prediction~\citep{goodfellow2014explaining,madry2017towards,carlini2017towards,ustun2019actionable,karimi2021algorithmic}. MSP imports this minimum-change logic into analyst-decision space: it holds data fixed and changes the declared analysis. Failure means loss of support for a causal claim, not an observation or prediction.

% \textbf{Simultaneous inference.} Simultaneous confidence regions are an inferential input to our population-level envelope, not a claim that one band construction is uniformly optimal. The present experiments compare Bonferroni and common-resample sup-\(t\) bands; a broader comparison with multiplier and cluster-aware constructions remains outside the current empirical scope.

\section{Methodology}

\subsection{Framework}
\label{sec:framework}
A specification space is built from \(K\) declared analyst decisions, or axes. Axis \(k\) has a finite option set \(L_k\), one baseline option \(s^0_k\), and a specification \(s=(s_1,\ldots,s_K)\) chooses one option per axis. The declaration also records each axis's \emph{type}, the structure its options are claimed to carry; the type plays no part in the distance and is used only to route a follow-up analysis (\Cref{sec:followup}). The feasible space \(S\subseteq\prod_{k=1}^K L_k\) contains \(s^0\) and may exclude incompatible combinations. We count changed decisions by
\[
c_0(s)=\#\{k:s_k\neq s^0_k\}.
\]
The binary grids used below are the special case \(L_k=\{0,1\}\), but the definition counts decisions rather than their coding.

\paragraph{Why one unit per decision.}
Unit cost is a declaration, not a claim that all decisions matter equally. Every result below holds for a prespecified weighted cost \(c_w(s)=\sum_k w_k\mathbf 1\{s_k\neq s^0_k\}\), \(w_k>0\) (the weights of \Cref{tab:prespec-template}), except \Cref{thm:certified}'s certificate and its corollary: the lower endpoint \(r+1\) uses the integrality of \(c_0\), whose weighted replacement needs costs the declaration lacks. We keep unit weights because weights must be fixed outside the data, letting a reader recost the enumeration without refitting (\Cref{app:empirics}).

The analyst declares a claim direction \(d\in\{-1,+1\}\) together with the space and baseline. For clarity take \(d=+1\); exchange endpoints for \(d=-1\). Let \(\mathrm{CI}(s)=[L(s),U(s)]\), and define
\[
\begin{aligned}
\mathcal N_{\mathrm{dir}}(S)&=\{s\in S:L(s)\le0\},\\
\mathcal N_{\mathrm{null}}(S)&=\{s\in S:0\in\mathrm{CI}(s)\},\\
\mathcal N_{\mathrm{rev}}(S)&=\{s\in S:U(s)<0\}.
\end{aligned}
\]
For any \(F\subseteq S\), write
\[
\mathrm{MSP}(S,F)=
\begin{cases}
\min_{s\in F}c_0(s),&F\ne\emptyset,\\
\infty,&F=\emptyset.
\end{cases}
\]
Call any \(s\in F\) attaining this minimum a \emph{witness}.
The reporting pair is
\[
\begin{aligned}
(\mathrm{MSP}_{\mathrm{null}},\mathrm{MSP}_{\mathrm{rev}})
&=\bigl(\mathrm{MSP}(S,\mathcal N_{\mathrm{null}}(S)),\\
&\qquad\mathrm{MSP}(S,\mathcal N_{\mathrm{rev}}(S))\bigr).
\end{aligned}
\]
% For compatibility with the standard shorthand, we also write
% \(\mathrm{MSP}=\mathrm{MSP}_{\mathrm{null}}\). 
Its directional minimum is
\(\mathrm{MSP}_{\mathrm{dir}}=\mathrm{MSP}(S,\mathcal N_{\mathrm{dir}}(S))\).
An unsupported baseline has \(\mathrm{MSP}_{\mathrm{dir}}=0\) by design.

An interval that does not support the claim either contains zero or lies
entirely on the other side of it, never both: since \(L(s)\le U(s)\),
\(L(s)\le0\) holds exactly when \(L(s)\le0\le U(s)\) or \(U(s)<0\). The
directional failure set therefore splits into the null-contact and reversal
sets, \(\mathcal N_{\mathrm{dir}}(S)=\mathcal N_{\mathrm{null}}(S)\,\dot\cup\,\mathcal
N_{\mathrm{rev}}(S)\), and the directional distance is the smaller of the
two distances:
\begin{equation}
\label{eq:dir-decomp}
\mathrm{MSP}_{\mathrm{dir}}=\min\{\mathrm{MSP}_{\mathrm{null}},\mathrm{MSP}_{\mathrm{rev}}\}.
\end{equation}

The pair distinguishes three exhaustive states: finite null contact; \emph{overshoot}, an infinite null-contact distance with a finite reversal witness; and \((\infty,\infty)\), universal directional support. Null-contact distance is infinite in the latter two and cannot distinguish them.

% \paragraph{Scope and reporting.}
% MSP is conditional on the declared specification space, baseline, and cost: use defensible axes, keep the estimand fixed within a grid, and treat estimand-changing options as separate facets rather than axes (\Cref{sec:card-krueger}). Weighted or ordered costs, continuous grids, and multi-baseline reporting with a prespecified lower envelope are all supported; see \Cref{app:empirics} for templates and \Cref{sec:discussion} for the reporting recommendation.

\subsection{What a declaration certifies}
The reported distance is conditional on its declaration, so we ask what an \emph{enlargement} implies: a richer declaration adding analyses (new options or axes) while leaving every declared analysis, distance, and result unchanged. Let \(S^\star\) be such a space, into which the declaration embeds via an injective map \(\iota:S\hookrightarrow S^\star\) fixing the baseline and preserving distance, \(c^\star_0(\iota(s))=c_0(s)\); an added axis maps each declared specification to its baseline option, treating option- and axis-enlargement alike. We write \(S\) for \(\iota(S)\) below. With \(r\ge0\) an integer and \(B_r=\{s\in S^\star:c^\star_0(s)\le r\}\) the specifications at most \(r\) decisions from the baseline, call \(S\) \emph{\(r\)-complete} in \(S^\star\) when \(B_r\subseteq S\): out to distance \(r\), the analyst declared everything the space admits. Every declaration is \(0\)-complete.

\begin{theorem}[Certified interval for the latent distance]
\label{thm:certified}
Let \((S,F)\) and \((S^\star,F^\star)\) satisfy (i) the distance-preserving embedding above, and (ii) \(F=F^\star\cap S\). If \(S\) is \(r\)-complete in \(S^\star\), then
\[
\min\{\mathrm{MSP}(S,F),\,r+1\}
\le\mathrm{MSP}(S^\star,F^\star)
\le\mathrm{MSP}(S,F),
\]
and the two bounds are sharp over the enlargements satisfying (i)--(ii): each is attained by some such enlargement, though not in general by a given one. The interval is a single point exactly when \(\mathrm{MSP}(S,F)\le r+1\). For finite \(\mathrm{MSP}(S,F)\), completeness at \(r=\mathrm{MSP}(S,F)-1\) is necessary as well as sufficient for exactness uniformly over failure sets consistent with the declaration (\Cref{app:proofs}).
\end{theorem}

\begin{corollary}[Distance one certifies itself]
\label{cor:distance-one}
If the baseline supports the declared claim and \(\mathrm{MSP}(S,F)=1\), then \(\mathrm{MSP}(S^\star,F^\star)=1\) for every \(S^\star\) satisfying (i)--(ii), with no completeness hypothesis.
\end{corollary}

An enlargement never raises the distance; \(r\)-completeness supplies the lower bound monotonicity misses, so a reported \(\infty\) only guarantees no analysis within \(r\) decisions of the baseline fails in the richer space. The certificate excludes changes to a specification's implementation, estimand, or admissibility, and intervals from a simultaneous band, whose width grows with \(\lvert S\rvert\) (\Cref{rem:sim-not-stable} in \cref{app:proofs}).

Completeness is auditable for declared options but only a scope statement for undeclared axes, and admits no approximate weakening (\Cref{rem:two-incompleteness,rem:no-approximate} in \cref{app:proofs}); an added non-baseline axis normally limits the certificate to radius zero (\Cref{rem:axis-reach}). Distance one is thus no convenient special case: it is the largest distance certified against every enlargement without any completeness statement (\Cref{cor:distance-one}).

For a product-form sub-declaration \(S_D\), built from option subsets \(D_k\ni s^0_k\) under the admissibility rule of \(S^\star\), let \(W_F=\{s\in F:c_0(s)=m^\star(F)\}\) be the nearest witnesses of \(F\subseteq S^\star\), where \(m^\star(F)=\mathrm{MSP}(S^\star,F)<\infty\). Then
\begin{equation}
  \label{eq:recovery}
  \mathrm{MSP}(S_D,F\cap S_D)=m^\star(F)
  \quad\Longleftrightarrow\quad
  S_D\cap W_F\ne\varnothing,
\end{equation}
so a sub-declaration overstates the reference distance exactly when its omissions hide every nearest witness (proof in \Cref{app:proofs}).
 
\subsection{From diagnostic to inference}
\label{sec:estimation}
Let \(\tau(s)\) be the population value targeted by specification \(s\), and set
\[
\begin{aligned}
\mathcal N^*_{\mathrm{dir}}(S)&=\{s\in S:\tau(s)\le0\},\\
\mathrm{MSP}^*_{\mathrm{dir}}&=\mathrm{MSP}(S,\mathcal N^*_{\mathrm{dir}}(S)).
\end{aligned}
\]
Inference targets this directional minimum, since the population analogue of null contact, \(\tau(s)=0\), is a knife-edge. We write \(\widehat{\mathrm{MSP}}_{\mathrm{dir}}\) for the sample \(\mathrm{MSP}_{\mathrm{dir}}\) of \Cref{sec:framework}.

\begin{proposition}[Directional finite-sample validity]
\label{prop:dir-validity}
If
\(\Pr\{L(s)\le\tau(s)\}\ge1-\alpha\) for every
\(s\in\mathcal N^*_{\mathrm{dir}}(S)\), then
\[
\Pr\{\widehat{\mathrm{MSP}}_{\mathrm{dir}}
\le\mathrm{MSP}^*_{\mathrm{dir}}\}\ge1-\alpha.
\]
\end{proposition}

The hypothesis is marginal coverage at the single non-random population minimizer, so pointwise intervals suffice. The sample directional distance is thus a lower confidence bound: a small value diagnoses but does not prove population fragility.

\begin{corollary}[Testing a declared robustness radius]
\label{cor:radius-test}
Fix a threshold \(m>0\). For any distribution in \(H_{0,m}:\mathrm{MSP}^*_{\mathrm{dir}}<m\) that satisfies the one-sided coverage hypothesis of \Cref{prop:dir-validity}, the rule rejecting when \(\widehat{\mathrm{MSP}}_{\mathrm{dir}}\ge m\) has rejection probability at most \(\alpha\).
\end{corollary}

With \(S_{<m}=\{s\in S:c_0(s)<m\}\), \(H_{0,m}\) is the union null \(\bigcup_{s\in S_{<m}}\{\tau(s)\le0\}\), and the rule rejects exactly when \(L(s)>0\) on all of \(S_{<m}\). It is therefore the intersection--union test \citep{berger1982multiparameter}, with the reported distance as its statistic: it needs no multiplicity correction, paying instead in power as more specifications fall within radius. At radius one it is the baseline claim's one-sided test. Every rejection is correct on the single coverage event behind \Cref{prop:dir-validity}, so testing all radii at once gives the confidence set \([\widehat{\mathrm{MSP}}_{\mathrm{dir}},\infty]\).

To bound both sides, let \(\ell(s)\le u(s)\) be the lower and upper limits of a simultaneous band over \(S\), and define the possible-loss and failure-side endpoints
\[
\begin{aligned}
\widehat M_{\mathrm{lo}}^{\mathrm{sim}}
&=\mathrm{MSP}(S,\{s:\ell(s)\le0\}),\\
\widehat M_{\mathrm{hi}}^{\mathrm{sim}}
&=\mathrm{MSP}(S,\{s:u(s)\le0\}).
\end{aligned}
\]

\begin{proposition}[Sharp endpoints of a simultaneous-confidence envelope]
\label{prop:sim-envelope}
If \(\Pr\{\ell(s)\le\tau(s)\le u(s)\ \forall s\in S\}\ge1-\alpha\), then
\[
\Pr\{\widehat M_{\mathrm{lo}}^{\mathrm{sim}}
\le\mathrm{MSP}^*_{\mathrm{dir}}
\le\widehat M_{\mathrm{hi}}^{\mathrm{sim}}\}\ge1-\alpha.
\]
Both endpoints are sharp for the realized band.
\end{proposition}

Diagnosing loss of support is thus cheaper than certifying population failure: a finite pointwise distance needs one lower limit below zero, while a finite failure-side endpoint needs an entire simultaneous interval there; if none qualifies, no failure is certified at any distance. All three results need intervals meeting their coverage hypotheses, which bootstrap intervals satisfy only approximately rather than exactly.

\section{Experiments}

\subsection{Simulation: overshoot is common under a weak effect}
\label{sec:simulation}
The controlled design has four binary axes: omit \(X_1\), omit \(X_2\), force linearity, and disable trimming. The baseline adjusts for confounders, uses a flexible model, and trims overlap observations; all \(16\) configurations use bootstrap CIs. Outcomes follow \(Y=\tau A+g(X)+\varepsilon\) under confounded treatment assignment (full construction in \Cref{app:empirics}; code available\footnote{\href{https://anonymous.4open.science/r/Minimum-Specification-Perturbation-Robustness-as-Distance-to-Falsification-in-Causal-Inference-B12B/README.md}https://anonymous.4open.science}).

At \(\tau=0.3\), omitting a confounder often drives a positive estimate past zero to a significantly negative one. Consequently, \(\mathrm{MSP}_{\mathrm{null}}=\infty\) in 81.5\% of replicates, even though a finite reversal witness exists in every such replicate (\Cref{tab:directional-summary} in \cref{app:empirics}). At \(\tau=1.5\), by contrast, the 16 infinite cases have no reversal witness and are robust within the grid. Thus the pair resolves an ambiguity that is common, not rare, in this weak-effect regime.

% The simulation does not claim that a dispersion summary should itself be a decision rule. It shows that, when such configuration summaries are thresholded for the same robustness task, they need not recover distance-to-falsification or distinguish overshoot from support. Operating characteristics for a distance-two rule and for those summaries, ROC curves, and additional stress tests for hidden bias, attrition, transportability, and interference are in \Cref{app:empirics}.

\subsection{LaLonde: an auditable witness, not a verdict}
\label{sec:lalonde}
We first use the NSW treated sample with CPS comparison-control data. The declared space has four axes: estimator (OLS/IPW), covariates (basic/full), form (linear/nonlinear), and trimming (off/on). Trimming is meaningful only on the IPW arm, leaving 12 admissible configurations (\Cref{fig:lalonde-curve}). The baseline is OLS with basic covariates, a linear model, and no trimming. It gives \(\hat\tau=-\$3{,}807\) with percentile-bootstrap 95\% CI \([-\$4{,}933,-\$2{,}654]\) from \(B=4{,}000\) resamples.

\begin{figure*}[t]
    \centering
    \includegraphics[width=\textwidth]{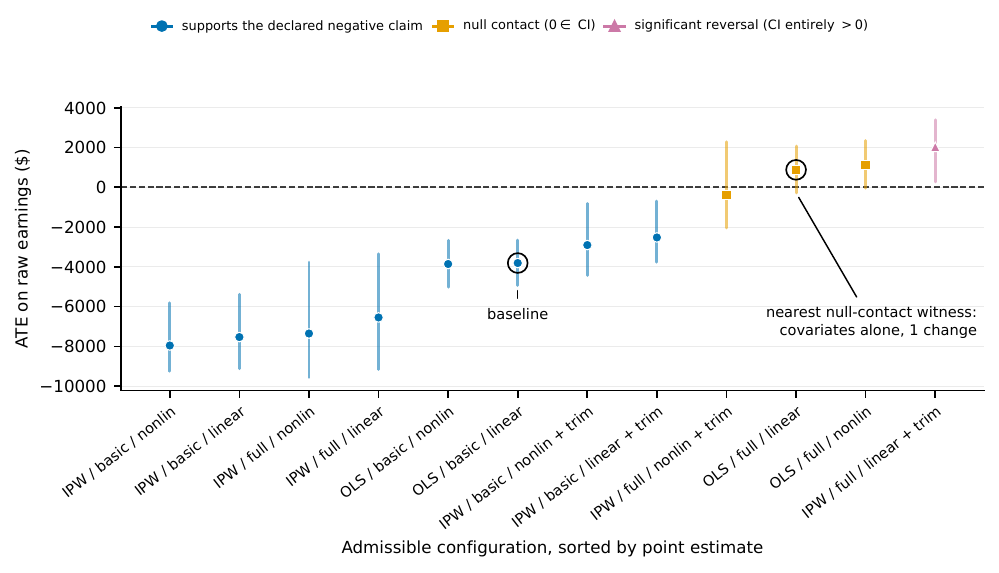}
    \caption{Eight of the twelve LaLonde configurations support the
    negative claim, yet the circled baseline is one decision from losing
    support. Changing only its covariate set gives the other circled
    interval, which contains zero (\(\mathrm{MSP}_{\mathrm{null}}=1\)). The
    single significantly positive interval is three decisions away
    (\(\mathrm{MSP}_{\mathrm{rev}}=3\)). Bars show 95\% percentile-bootstrap
    confidence intervals.}
    \label{fig:lalonde-curve}
\end{figure*}

Switching from basic to full covariates gives \(\hat\tau=\$879\) and CI \([ -\$280,\$2{,}077]\), so \(\mathrm{MSP}_{\mathrm{null}}=1\). A significant reversal occurs at distance three (\(\mathrm{MSP}_{\mathrm{rev}}=3\)). This pair audits the reported result's reliance on covariate adjustment, not equal credibility between basic and full adjustment. If full adjustment is preferable, the distance-one witness reflects bias correction, not causal fragility. MSP leaves that to the design argument.

\paragraph{Descriptive noise-probe calibration.}
Declared noise probes, synthetic axes irrelevant in expectation, give the pair a reference scale (\Cref{fig:noise-probes} in \cref{app:empirics}). The calibration is post hoc: designed after the pair was known but frozen before execution; it attaches no error rate. Each block probe adds a jointly row-permuted copy of the four covariates that the witness adds, matching it in size. Each added singly at the baseline, fifty block probes shift \(U(s)\) by at most \$31. Fifty single Gaussian covariates shift it by at most \$15, against \$4{,}732 for the witness. With two block probes added as axes, the pair stays \((1,3)\). No minimal witness contains a probe, and no probe-only perturbation loses support.

The simultaneous analysis makes the diagnostic/inferential distinction concrete. For a declaration \(S\), the common-resample sup-\(t\) critical value is
\[
\begin{aligned}
 Z_b&=\max_{s\in S}
 \left|\frac{\widehat\tau_b(s)-\widehat\tau(s)}
 {\widehat{\operatorname{se}}(s)}\right|,\\
 q_S&=\operatorname{quantile}_{0.95}
 \{Z_b:b=1,\ldots,B\}.
\end{aligned}
\]
where the standard errors are bootstrap standard deviations. On the 12-configuration grid, \(q_S=3.16418\). The pointwise pair is \((1,3)\), whereas both the sup-\(t\) envelope (\(q=3.16\)) and the Bonferroni envelope are \([1,\infty]\). One change loses simultaneous directional support, but no configuration's band lies entirely positive. The distance-three reversal is therefore not confirmed at the population level. This compares realized band widths, not empirical coverage, since LaLonde supplies no known population value.

\paragraph{The pair is conditional on the declaration.}
Holding the four axes, the baseline, and the admissibility rule fixed, we widen only the option sets. Estimator options: OLS, H\'ajek IPW, or Horvitz--Thompson (HT) IPW; covariates: basic, full, or full with pairwise products of continuous covariates. Functional form: linear or quadratic; trimming: none, 1\%, 5\%, or 10\%, giving \(\lvert S\rvert=54\). The binary grid's twelve configurations reappear with identical intervals at unchanged distance. So the embedding and failure-set conditions (i)--(ii) of \Cref{thm:certified} hold by construction (\Cref{app:empirics}). The null-contact distance stays at one, but the reversal distance falls from three to one. Switching only the covariate set, from basic to full with pairwise products, gives \(\hat\tau=\$1{,}205\) with CI \([\$37,\$2{,}418]\).

The binary pair was therefore an upper bound, not an error. Since every declaration is complete to radius zero, the binary grid alone certifies \(\mathrm{MSP}_{\mathrm{null}}=1\) exactly (\Cref{cor:distance-one}). It also brackets \(\mathrm{MSP}_{\mathrm{rev}}\) in \([1,3]\) (\Cref{thm:certified}); enlargement landed at the bottom of that bracket. The enriched pair \((1,1)\) is in turn exact in both coordinates against every further enlargement satisfying (i)--(ii), including added axes. So the regress of each richer declaration undercutting the last stops once a coordinate reaches one, without claiming completeness. Stress-testing every product-form sub-declaration shows when a thinner declaration would hide that witness. Four-configuration sub-declarations recover the reference directional distance in 41.7--71.4\% of cases and report a false \((\infty,\infty)\) in up to 58.3\%. At twelve configurations, by contrast, recovery is 82.9--100\%, with no false cases (\Cref{tab:declaration-stress}).

\paragraph{The diagnosis survives a published pipeline.}
The preceding grid is deliberately generic. To ask whether the diagnosis survives an applied workflow's explicit choices, we reanalyse the randomized NSW sample using the \citet{imbens2015matching} pipeline (construction details in \Cref{app:empirics}). Within each 12-configuration facet, matching the paper's selected propensity-score model, three subclasses, and all-covariate adjustment, \(B=1{,}000\) stratified resamples give an untrimmed baseline of \(\$1{,}610\), CI \([\$459,\$3{,}201]\). Changing only the estimator to 1:1 propensity-score matching gives \(\$374\), CI \([-\$1{,}121,\$4{,}138]\), so \((\mathrm{MSP}_{\mathrm{null}},\mathrm{MSP}_{\mathrm{rev}})=(1,\infty)\) within that facet. The CHIM-target facet yields the same pair, with matching again the one-move null witness. This is a declared reimplementation, not a claim that the article enumerated every combination.

\paragraph{Baseline-recentered sensitivity.}
\label{par:baseline-sensitivity}
To assess baseline dependence, we recenter the enriched 54-configuration declaration at four plausible baselines that each support the negative claim. The primary baseline remains OLS with basic covariates, a linear form, and no trimming. The alternatives vary the functional form or the IPW estimator, while leaving the outcome scale and target fixed. Distances count decisions, not encoding bits (\Cref{prop:encoding}, \cref{app:proofs}).

\begin{table}[t]
    \centering
    \setlength{\tabcolsep}{2pt}
    \begin{tabular}{@{}p{0.34\linewidth}p{0.30\linewidth}p{0.13\linewidth}c@{}}
        \toprule
        baseline & 95\% CI & MSP pair & \(\mathrm{MSP}_{\mathrm{dir}}\) \\
        \midrule
        OLS, basic, linear, no trim & \([-4933,-2654]\) & \((1,1)\) & 1 \\
        OLS, basic, quadratic, no trim & \([-5017,-2666]\) & \((1,2)\) & 1 \\
        H\'ajek IPW, basic, linear, no trim & \([-9120,-5384]\) & \((2,2)\) & 2 \\
        HT IPW, basic, linear, no trim & \([-11191,-8106]\) & \((1,2)\) & 1 \\
        \midrule
        lower envelope & & \((1,1)\) & 1 \\
        \bottomrule
    \end{tabular}
    \caption{Baseline-specific MSP pairs for enriched LaLonde (\(|S|=54\), \(B=4{,}000\)). All supporting baselines lose directional support within two decisions; the lower envelope is descriptive only.}
    \label{tab:baseline-sensitivity}
\end{table}

\Cref{tab:baseline-sensitivity} collects the four recentered pairs. The primary conclusion therefore does not depend on a unique supporting baseline. The lower-envelope directional distance is one. The table also shows why baseline-specific pairs should remain the primary report. The H\'ajek IPW row loses support only at distance two, because the baseline sets the reference point from which analyst-decision distance gets measured.

% cân nhắc bỏ đi
% \paragraph{Coding the same options in binary.}
% The 54 analyses can equally be described by binary indicators, one-hot on the unordered estimator axis, thermometer on the three ordered ones, reference options fixed in advance. Nothing about the analyses changes, but the distances do: the reversal witness lies one decision from the baseline and two bits from it, so the binary reading reports \(\mathrm{MSP}^{\beta}_{\mathrm{rev}}=2\) where the decision count reports \(1\). The inequality is the one \Cref{prop:encoding} predicts, and it shows that the inflation is not hypothetical. Counting decisions, not bits, is what makes the reported distance a statement about the analyst's choices.

\paragraph{Data fragility.}
% On the NSW randomized sample (185 treated and 260 controls), a four-axis grid gives
% \((\mathrm{MSP}_{\mathrm{null}},\mathrm{MSP}_{\mathrm{rev}})=(\infty,\infty)\): all 16 intervals support the declared positive claim. With 200 complete-randomization permutations and fixed bootstrap resamples, only 2 permutations attain \(\mathrm{MSP}=\infty\), giving add-one \(\hat p\approx0.015\). We report this as a descriptive randomization reference, not as a formal calibration result.

On the same randomized NSW sample, a separate four-axis grid of 16 configurations gives \((\mathrm{MSP}_{\mathrm{null}},\mathrm{MSP}_{\mathrm{rev}})=(\infty,\infty)\). Every interval supports the declared positive claim, on each outcome scale separately (\Cref{app:empirics}). That result is data-fragile under a different perturbation, as \Cref{tab:fragility-vs-msp} records: an adversarial continuous-outcome Fragility Index is 3. Conversely, the two synthetic regimes in \Cref{app:empirics} include a tail-driven setting (median FI 7, \(P(\mathrm{MSP}_{\mathrm{null}}=\infty)=0.958\)) and a specification-sensitive setting (median FI 38, \(P(\mathrm{MSP}_{\mathrm{null}}=1)=0.925\)). FI and MSP are non-equivalent, diverging in both directions.

\begin{table}[t]
    \centering
    \setlength{\tabcolsep}{3pt}
    \begin{tabular}{@{}lcc@{}}
        \toprule
        metric & value & reading \\
        \midrule
        FI (adversarial) & 3 & data-fragile \\
        FI (random median) & 16 & data-fragile \\
        MSP pair & \((\infty,\infty)\) & specification-robust \\
        \bottomrule
    \end{tabular}
    \caption{Randomized NSW shows distinct vulnerabilities: three adversarial reassigned outcomes erase support, whereas all 16 declared analytic configurations retain full support.}
    \label{tab:fragility-vs-msp}
\end{table}

\subsection{Card--Krueger: faceting an estimand-changing axis}
\label{sec:card-krueger}
We next revisit the Card--Krueger fast-food difference-in-differences design, whose claim is positive. Estimator (first-difference DiD versus long OLS DiD) and controls (off/on) are two of three analyst moves at a fixed estimand. Sample (full, with permanent closures coded as zero, versus balanced panel) is the third, forming a \(K=3\) grid. Outcome scale changes the estimand, so we treat level employment and \(\log(1+\text{employment})\) as separate facets, not a single scale-switch move.

\begin{table}[t]
    \centering
    \setlength{\tabcolsep}{3pt}
    \begin{tabular}{@{}lcc@{}}
        \toprule
        quantity & levels & \(\log(1+\cdot)\) \\
        \midrule
        baseline \(\hat\tau\) & 2.785 & 0.075 \\
        baseline 95\% CI & [0.235, 5.396] & [-0.043, 0.200] \\
        \(\mathrm{MSP}_{\mathrm{null}}\) & \(\infty\) & 0 \\
        \(\mathrm{MSP}_{\mathrm{rev}}\) & \(\infty\) & \(\infty\) \\
        \bottomrule
    \end{tabular}
    \caption{Faceted Card--Krueger MSP results (\(K=3\), eight configurations, \(B=4{,}000\)): every level-scale configuration retains support, while the log-scale baseline already contains zero.}
    \label{tab:card-krueger-msp}
\end{table}

\Cref{tab:card-krueger-msp} reports the two facets side by side. On the level scale all eight intervals lie entirely above zero, the nearest of them at \([0.249,5.389]\), so the pair is \((\infty,\infty)\). On the log scale the declared baseline itself contains zero and the distance is zero. The two facets therefore land at opposite ends of the reporting scale, universal support against an unsupported baseline, and they already disagree at the baseline, before any perturbation. Pooling the scales would obscure both facts by treating an estimand change as an ordinary axis flip. This example does not adjudicate the Card--Krueger debate; it demonstrates the design discipline that a specification distance requires.

The level-scale facet is also where the resample count matters. Its nearest configuration has a lower limit of \(0.249\) at \(B=4{,}000\) but \(-0.007\) at \(B=300\), where a \(2.5\%\) percentile falls between two order statistics. A status decided by the sign of a limit that close to zero is a statement about bootstrap resolution rather than about the declaration, so we report this grid at the larger \(B=4{,}000\) resample count.

\paragraph{Published specification tests as an MSP grid.}
\citet{card1993minimum} published robustness tables varying the wage-gap versus New-Jersey contrast, FTE construction, and several further controls (full list in \Cref{app:empirics}). We evaluate the feasible product of those one-at-a-time checks, treating FTE construction and the exposure contrast as facets since they change the outcome scale. In the primary FTE facet (managers included, part-time weight \(0.5\)), the 64-configuration baseline reproduces their published estimate to printed precision, \(14.916\), CI \([3.034,27.540]\) at \(B=1{,}000\). Adding regional controls alone produces \(11.979\), CI \([-1.666,25.831]\), so \((\mathrm{MSP}_{\mathrm{null}},\mathrm{MSP}_{\mathrm{rev}})=(1,\infty)\) --- the same null-contact distance holds across all three FTE operationalizations. The New-Jersey-dummy facet instead begins at \(2.304\), CI \([-0.572,5.163]\), so \(\mathrm{MSP}_{\mathrm{null}}=0\): the two facets disagree even at the baseline. These \(\infty\) entries are conditional on the declared finite grids, newly re-estimated from raw data (\Cref{app:empirics}).

\section{From a witness to a follow-up analysis}
\label{sec:followup}

MSP names the nearest specifications that withdraw support and routes them to
a follow-up: a reporting step, not a new sensitivity measure. MSP
locates where to look; an existing operation examines further.

\paragraph{Routing the witness.}
Each declared axis carries a type---\emph{ordered}, \emph{nested}, or
\emph{categorical}---fixed in advance, not from data. Every nearest witness is
retained: an axis changed by all is load-bearing, by only some an alternative
route. The type determines the follow-up (\Cref{fig:followup-protocol}), full
taxonomy in \Cref{app:witness-followup}.

\begin{figure}[H]
\centering
\fbox{%
\begin{minipage}{0.93\linewidth}

\textbf{Witness-guided follow-up}

\textbf{Input:} the declared space, baseline, claim direction, and each axis's type.

\begin{enumerate}
\setlength{\itemsep}{1pt}
\setlength{\parskip}{0pt}
\item Compute the MSP pair; retain all nearest witnesses.
\item Collect their changed axes; mark load-bearing vs.\ alternative routes.
\item Route each axis by type: ordered \(\rightarrow\) search for first loss
      of support; nested \(\rightarrow\) find minimal failing component;
      categorical \(\rightarrow\) report a structured contrast.
\item Report the original pair and follow-up apart.
\end{enumerate}
\end{minipage}}
\caption{Locates and routes a witness; computes no new measure, adjudicates no
specification.}
\label{fig:followup-protocol}
\end{figure}

\paragraph{What the follow-up changes.}
Ordered and nested follow-ups may evaluate analyses outside the declaration;
\Cref{cor:distance-one} preserves a fixed distance-one coordinate under
enlargement. Larger-distance refinements belong to a new declaration, not an
updated MSP. A categorical axis supplies no ordering, so its follow-up only
contrasts implementations; the reported pair stays the original declaration's
output.

\paragraph{A categorical witness reveals a hidden implementation decision.}
In the published-pipeline reanalysis of \Cref{sec:lalonde}, the nearest
null-contact witness changes three-subclass adjustment to \(1{:}1\)
propensity-score matching \citep{imbens2015matching}, taking the first control
among exact ties. At \(B=1{,}000\) stratified draws, subclassification gives
\(\hat\tau=1610\), CI \([459,3201]\). Matching instead gives \(374\), CI
\([-1121,4138]\).

The structured follow-up identifies the mechanism: of 185 treated units, 108
have multiple exactly nearest controls, and the first-tie implementation
weights only 59 of 260, cutting effective control sample size from 206.4 to
17.4 and placing 46.0\% of weight on five observations. Dividing each unit's
weight equally across all its nearest controls instead is outcome-blind, keeps
the same model, adjustment, sample, and bootstrap draws, uses 209 controls,
raises effective sample size to 120.6, and gives \(\hat\tau=1785\), CI
\([157,3230]\): directional support returns. The witness named ``matching''
thus traces to tie handling and weight concentration, not estimator choice,
and does not recompute MSP: the label concealed an implementation decision
future declarations should specify. The \(B=4{,}000\) check, random-tie
analysis, and nested-axis example are in \Cref{app:witness-followup}.

\section{Discussion}
\label{sec:discussion}

\paragraph{What the evidence answers.}
We asked how many declared analyst decisions must change before the data stop supporting a causal claim. The pair \((\mathrm{MSP}_{\mathrm{null}},\mathrm{MSP}_{\mathrm{rev}})\) answers that, and its minimum \(\mathrm{MSP}_{\mathrm{dir}}\) compresses it to one number. One covariate decision sufficed on the LaLonde raw-earnings grid, one control-set decision on the Card--Krueger grid. Its faceted level-scale version retained support throughout. In simulation, the median distance at \(\tau=0.3\) was one decision. Dispersion moved by only 0.005 across four regimes; the pair still distinguished weak- from high-effect regimes (\Cref{tab:directional-summary}). A null-forcing test~\citep{simonsohn2020specification} answers the no-effect question (\Cref{tab:sca-vs-msp}).

\paragraph{Consequences for reporting.}
Four items make a reported distance interpretable: the baseline it is measured from; the declaration, axes, options, and admissibility rule fixing it (\Cref{tab:prespec-template}); the pair, not null contact alone, which assigns \(\infty\) to universal support and reversal alike; and the exhausted radius, which certifies the distance. Where several baselines are defensible, report each one separately (\Cref{tab:baseline-sensitivity}).

\paragraph{Reading the number.}
\(\mathrm{MSP}_{\mathrm{dir}}=m\) says every analysis within
\(m-1\) decisions keeps support and some analysis at \(m\) does not.
We follow \citet{vanderweele2017sensitivity} in declining a
universal cutoff. Interpret a value using the declaration's exhausted radius
and nearest-witness count. In enriched LaLonde, two of the five
admissible analyses withdraw support at the primary baseline, and
ten do the same at distance two for H\'ajek.

\paragraph{Limitations.}
Exhaustive enumeration scales exponentially in decision axes. Outward search can certify small distances without enumerating it (\Cref{cor:anytime} in \cref{app:proofs}). MSP is conditional on the declared analysis space: it assesses specified choices, not ones never declared. Pointwise readings are diagnostic, not inferential; simultaneous bands provide joint coverage at a precision cost. The noise-probe calibration is descriptive. The directional decision rule \(\mathrm{MSP}_{\mathrm{dir}}\ge2\) of \Cref{tab:decision-rules} is stringent, powerless at the moderate effect \(\tau=0.7\) in this design.
%đoạn này cần thiết không?
% \paragraph{Where this leads.}
% Two extensions follow directly from these limits. \Cref{cor:anytime} makes exact computation cheap when the distance is small, but a space too large to enumerate and a distance that is genuinely large still call for search over the declared graph, at the cost of an approximation guarantee the present exact computation does not need. Comparing data perturbation and specification perturbation on a common design would sharpen the divergence we document here into a statement about when each measure is the informative one. Beyond that, the declared cost \(c_0\) treats all analyst decisions as equally weighted; ordered and weighted costs are defined in \Cref{app:empirics} but require substantive prespecification that we have not attempted to standardize.

\section{Conclusion}
MSP adds an auditable distance dimension to robustness analysis in a declared specification space. The pair \((\mathrm{MSP}_{\mathrm{null}},\mathrm{MSP}_{\mathrm{rev}})\) distinguishes overshoot from genuine universal support, while \(\mathrm{MSP}_{\mathrm{dir}}\) summarizes distance to loss of support. Its two readings differ: pointwise intervals diagnose where support ends, while a simultaneous band certifies that distance population-wide. Diagnosing fragility is therefore cheaper than certifying it; our applications report a finite distance alongside an unbounded envelope. Beyond \emph{where do the estimates lie}, a declared baseline and cost let that enumeration answer \emph{which decision, and how far}. Examples include covariate adjustment in LaLonde and the regional-control set in Card--Krueger, both more contestable than a summary statistic. The pair also diverges from hidden-bias and data-fragility measures, making explicit a dependence that the estimate distribution otherwise leaves unstated.

\begingroup
\small
\setlength{\bibsep}{1pt}
\bibliography{references}
\endgroup

\clearpage
\appendix
\thispagestyle{empty}
\onecolumn
\aistatstitle{Minimum Specification Perturbation: Robustness as Distance-to-Falsification in Causal Inference: \\
Supplementary Materials}
\section{Proofs}
\label[appendix]{app:proofs}

\subsection*{Encoding}

\begin{proposition}[Relabeling and re-encoding]
\label{prop:encoding}
Let \(S\subseteq\prod_kL_k\) have baseline \(s^0\). (a) A coordinatewise bijective relabeling that preserves the baseline, estimates, and intervals leaves MSP unchanged. (b) Let \(\beta(s)=(\beta_1(s_1),\ldots,\beta_K(s_K))\) encode each axis separately and injectively into binary blocks, the block for axis \(k\) having a fixed length that does not depend on the option encoded. If \(c^\beta(s)=d_H(\beta(s),\beta(s^0))\) and \(\mathrm{MSP}^\beta\) is computed on the code image \(\beta(S)\) with the same feasible set, then
\[
c^\beta(s)\ge c_0(s)\quad\text{for all }s\in S,
\qquad
\mathrm{MSP}^\beta\ge\mathrm{MSP},
\]
with strict inequality possible.
\end{proposition}

\subsection*{Proofs of the main-text propositions}

Proofs are provided in full detail. For self-containedness we restate the key definitions and proposition statements before each proof.

\paragraph{Notation and definitions.}
Axis \(k\) carries a finite declared option set \(L_k\), \(m_k=|L_k|\ge2\), with a declared baseline option \(s^0_k\in L_k\). The \emph{specification space} is \(S\subseteq\prod_{k=1}^K L_k\) with \(s^0\in S\), the inclusion possibly strict when some combinations are inadmissible. The cost of \(s\) is its Hamming distance to the baseline,
\[
c_0(s)=\#\{k:\ s_k\neq s^0_k\}.
\]
The binary case is \(L_k=\{0,1\}\), \(s^0=\mathbf 0\), where \(c_0(s)=\|s\|_0=\sum_k s_k\). Let \(\hat{\tau}(s)\) be an estimated treatment effect and let \(\mathrm{CI}(s)=[L(s),U(s)]\), with \(L(s)\le U(s)\), be its CI. The \emph{set of specifications whose CIs contain zero} is
\[
\mathcal{N}_{\mathrm{null}}(S) = \{s \in S : 0 \in \mathrm{CI}(s)\}.
\]
For theory it is useful to separate the combinatorial object from its causal interpretation. Let \(F \subseteq S\) be any feasible set. The \emph{Minimum Specification Perturbation} (MSP) of \((S,F)\) is
\[
\mathrm{MSP}(S,F) =
\begin{cases}
\min \{c_0(s) : s \in F\}, & F \neq \emptyset, \\
\infty, & F = \emptyset.
\end{cases}
\]
The causal version sets \(F = \mathcal{N}_{\mathrm{null}}(S)\).

\paragraph{Directional notation.}
Let \(d=+1\) be the declared direction, fixed before the data are seen (the case \(d=-1\) follows by exchanging the endpoints). Define
\[
\begin{aligned}
\mathcal N_{\mathrm{dir}}(S)&=\{s\in S:L(s)\le 0\},\\
\mathcal N_{\mathrm{rev}}(S)&=\{s\in S:U(s)<0\},
\end{aligned}
\]
with \(\mathrm{MSP}_{\mathrm{dir}}=\mathrm{MSP}(S,\mathcal N_{\mathrm{dir}}(S))\) and \(\mathrm{MSP}_{\mathrm{rev}}=\mathrm{MSP}(S,\mathcal N_{\mathrm{rev}}(S))\). These are sample objects, computed from the realized intervals; where the contrast with their population counterparts matters we write \(\widehat{\mathcal N}_{\mathrm{dir}}(S)\) and \(\widehat{\mathrm{MSP}}_{\mathrm{dir}}\) for the same two quantities. We use throughout that \(F\subseteq F'\) implies \(\mathrm{MSP}(S,F')\le\mathrm{MSP}(S,F)\), and that \(\mathrm{MSP}(S,A\cup B)=\min\{\mathrm{MSP}(S,A),\mathrm{MSP}(S,B)\}\), both immediate from the defining minimum and the convention \(\min\emptyset=\infty\).

\paragraph{Directional decomposition \eqref{eq:dir-decomp}.}
\emph{Statement.} \(\mathcal N_{\mathrm{dir}}(S)=\mathcal N_{\mathrm{null}}(S)\,\dot\cup\,\mathcal N_{\mathrm{rev}}(S)\), and hence \(\mathrm{MSP}_{\mathrm{dir}}=\min\{\mathrm{MSP},\mathrm{MSP}_{\mathrm{rev}}\}\le\mathrm{MSP}\).

\begin{proof}
If \(s\in\mathcal N_{\mathrm{null}}(S)\) then \(L(s)\le 0\le U(s)\), so \(L(s)\le0\) and \(s\in\mathcal N_{\mathrm{dir}}(S)\). If \(s\in\mathcal N_{\mathrm{rev}}(S)\) then \(L(s)\le U(s)<0\), so again \(s\in\mathcal N_{\mathrm{dir}}(S)\). Conversely let \(s\in\mathcal N_{\mathrm{dir}}(S)\), so \(L(s)\le0\). Either \(U(s)\ge0\), and then \(L(s)\le0\le U(s)\) gives \(s\in\mathcal N_{\mathrm{null}}(S)\); or \(U(s)<0\), and then \(s\in\mathcal N_{\mathrm{rev}}(S)\). The two cases are mutually exclusive, since membership in \(\mathcal N_{\mathrm{null}}(S)\) requires \(U(s)\ge0\) and membership in \(\mathcal N_{\mathrm{rev}}(S)\) requires \(U(s)<0\). This proves the disjoint decomposition. Taking \(\mathrm{MSP}(S,\cdot)\) of both sides and using the union rule gives \(\mathrm{MSP}_{\mathrm{dir}}=\min\{\mathrm{MSP},\mathrm{MSP}_{\mathrm{rev}}\}\), and \(\mathcal N_{\mathrm{null}}(S)\subseteq\mathcal N_{\mathrm{dir}}(S)\) gives \(\mathrm{MSP}_{\mathrm{dir}}\le\mathrm{MSP}\).
\end{proof}

\paragraph{\Cref{thm:certified} (Certified interval for the latent distance).}
\emph{Statement.} Suppose (i) there is an injective \(\iota:S\hookrightarrow S^\star\) with \(\iota(s^0)=s^{\star0}\) and \(c^\star_0(\iota(s))=c_0(s)\) for every \(s\in S\), and (ii) \(F=F^\star\cap\iota(S)\), where \(F\) is identified with \(\iota(F)\). Identifying \(S\) with \(\iota(S)\), write \(B_r=\{s\in S^\star:c^\star_0(s)\le r\}\) for an integer \(r\ge0\) and suppose \(S\) is \(r\)-complete in \(S^\star\), that is, \(B_r\subseteq S\). Then
\[
\min\{\mathrm{MSP}(S,F),\,r+1\}
\le\mathrm{MSP}(S^\star,F^\star)
\le\mathrm{MSP}(S,F).
\]

\begin{proof}
Write \(M=\mathrm{MSP}(S,F)\) and \(M^\star=\mathrm{MSP}(S^\star,F^\star)\).

\emph{Upper bound.} If \(F=\emptyset\) then \(M=\infty\) and the inequality is immediate. Otherwise let \(s^\dagger\in F\) attain \(M=c_0(s^\dagger)\). By (ii), \(F\subseteq F^\star\), so \(s^\dagger\) is one of the specifications competing in the minimum defining \(M^\star\); by (i) it enters that minimum at cost \(c^\star_0(s^\dagger)=c_0(s^\dagger)\). Hence \(M^\star\le M\).

\emph{Lower bound.} Suppose \(M^\star<\min\{M,r+1\}\). Then \(M^\star\) is finite; let \(t\in F^\star\) attain \(M^\star=c^\star_0(t)\). Both \(M^\star\) and \(r+1\) are integers, so the strict inequality gives \(M^\star\le\min\{M,r+1\}-1\le r\) and hence \(t\in B_r\); this step is where the integrality of the cost and of the radius is used, and it is why the lower endpoint has no \(r+1\) form for a non-integer radius or a cost that is not integer-valued. For a weighted cost, the corresponding replacement is the smallest cost attainable in \(S^\star\) that exceeds \(r\). By \(r\)-completeness \(t\in S\), and (ii) then gives \(t\in F^\star\cap S=F\). Since \(t\in S\), its cost is \(c_0(t)=c^\star_0(t)=M^\star\), so \(t\) competes in the minimum defining \(M\) and \(M\le M^\star\). This contradicts \(M^\star<M\). No case distinction on the size of \(M\) is needed, and the argument is unaffected by \(M=\infty\).

\emph{Sharpness.} The bounds are sharp over the class of enlargements satisfying (i)--(ii), not attained by every member of it. The upper endpoint is attained by \(S^\star=S\), \(F^\star=F\), which is \(r\)-complete for every \(r\). For the lower endpoint, assume \(M>r+1\), otherwise the two endpoints coincide and there is nothing to attain. A specification \(w\) at distance exactly \(r+1\) from the baseline and outside \(S\) is always available, because an enlargement may add axes as well as options: append \(r+1\) axes, each carrying one non-baseline option, embed \(S\) by sending every declared specification to the baseline option on each of them, and let \(w\) take the non-baseline option on all \(r+1\). When the ambient product of the declared axes already contains such a point, which requires \(r+1\le K\) and fails for a full-product declaration, it serves equally well. Take \(S^\star=S\cup\{w\}\) and \(F^\star=F\cup\{w\}\). Then \(B_r=B_r\cap S\subseteq S\), so \(S\) is \(r\)-complete; \(F^\star\cap S=F\), so (ii) holds; and \(M^\star=r+1\). For a \emph{given} \(S^\star\) neither endpoint need be attained: the latent distance may fall strictly inside the interval, which is what makes the interval a bracket rather than a prediction.

\emph{Necessity.} Assume \(M<\infty\), so that \(M-1\) is defined and the interval can degenerate; when \(M=\infty\) the interval is \([r+1,\infty]\) and never degenerates, so no necessity statement is available. Suppose some \(t\in S^\star\setminus S\) has \(c^\star_0(t)\le M-1\). Then \(F^\star=F\cup\{t\}\) satisfies (ii), and \(M^\star\le c^\star_0(t)<M\). Hence \(M^\star=M\) holds \emph{uniformly over the failure sets \(F^\star\) with \(F^\star\cap S=F\)} if and only if \(S\) is \((M-1)\)-complete in \(S^\star\). The quantifier matters: for a particular \(F^\star\) the equality can hold without completeness, since the undeclared near specifications may simply fail to be witnesses. What completeness rules out is the adversarial case, and it is the only hypothesis that does so without a model for \(F^\star\).

Since \(B_0=\{s^0\}\subseteq S\) always, every declaration is \(0\)-complete, which is \Cref{cor:distance-one}: when \(M=1\) the required radius is \(0\) and the hypothesis is free. Directly, \(M^\star\le1\) by the upper bound, and \(M^\star=0\) would put \(s^0\in F^\star\cap S=F\) and hence \(M=0\). The corollary's hypothesis that the baseline supports the claim is therefore implied by \(M=1\) rather than assumed on top of it.
\end{proof}

\paragraph{Nearest-witness recovery \eqref{eq:recovery}.}
\emph{Proof.} Refinement monotonicity gives only
\(\mathrm{MSP}(S_D,F\cap S_D)\ge m^\star(F)\). If the two values are equal,
a minimizer of the sub-declaration attains the reference distance and hence
lies in \(W_F\). Conversely, a retained nearest witness competes in the
minimum at cost \(m^\star(F)\), and nothing in \(F\cap S_D\subseteq F\) can
undercut that cost. If \(R_{k,\ell}\) indicates that option \(\ell\) is
retained on axis \(k\), the recovery condition has the monotone DNF
(disjunctive normal form) representation
\[
  \bigvee_{s\in W_F}
  \ \bigwedge_{k:s_k\ne s_k^0}R_{k,s_k}.
\]
Equivalently, an omitted-option set blocks recovery exactly when it is a
hitting set for the changed-option requirements of all nearest witnesses.
This is a deterministic identity relative to \(S^\star\), not a probability
model and not a claim that \(S^\star\) is a ground-truth universe.

\begin{remark}[Reach against added axes]
\label{rem:axis-reach}
Completeness at radius \(r\ge1\) is a strong requirement against axis enlargement, and essentially never holds: if \(S^\star\) appends an axis carrying a non-baseline option \(\ell\) that is admissible with every other axis at its baseline, then \(\iota(s^0)\) altered to \(\ell\) on that axis alone lies in \(B_1\setminus\iota(S)\), so \(S\) is \(0\)-complete and no more. Only a degenerate axis, none of whose non-baseline options is admissible at the baseline, escapes this. The certificate's reach against an undeclared axis is therefore exactly \Cref{cor:distance-one}. Against richer option sets on the declared axes the radius is whatever the declaration exhausts, and can exceed one. This asymmetry is not a defect of the bound but the reason a reported distance of one occupies a different epistemic position from a reported distance of two.
\end{remark}

\begin{corollary}[Anytime certification by ball enumeration]
\label{cor:anytime}
Evaluate the specifications of \(S\) in order of distance, exhausting the declared option sets at each radius, and stop at the first radius \(m\) that contains a member of \(F\). Relative to any \(S^\star\) with the same axes and the same option sets, the enumeration is option-complete to radius \(m-1\) by construction; a richer option set on a declared axis is not exhausted by it and must be met by the completeness hypothesis instead, and it certifies nothing about axes the declaration omits (\Cref{rem:two-incompleteness}). Given that, \(\mathrm{MSP}(S^\star,F^\star)=m\) for every enlargement \((S^\star,F^\star)\) satisfying (i)--(ii) of \Cref{thm:certified} in which \(S\) is \((m-1)\)-complete; for \(m=0\) no completeness hypothesis is involved. When the procedure stops at \(m<K\), write \(\bar B_m=\{s\in\prod_{k=1}^K L_k:c_0(s)\le m\}\). The number of specifications evaluated is
\[
\lvert S\cap\bar B_m\rvert\le\lvert\bar B_m\rvert
\le\sum_{j=0}^{m}\binom{K}{j}(L_{\max}-1)^j,
\]
polynomial in \(K\) and \(L_{\max}\) for fixed \(m\), and the specifications beyond radius \(m\) are never evaluated. The saving is conditional on stopping early: if \(F=\emptyset\) the procedure exhausts \(S\), returns \(\infty\), and costs \(\lvert S\rvert\), and a genuinely large distance costs accordingly.
\end{corollary}

\begin{proof}
Stopping at the first radius containing a witness gives \(\mathrm{MSP}(S,F)=m\). For \(m\ge1\), apply \Cref{thm:certified} with \(r=m-1\), for which \(\min\{m,r+1\}=m\); the theorem's hypotheses (i)--(ii) are carried by the statement, and without (ii) in particular the conclusion fails, since a set \(F^\star\) meeting \(S\) in more than \(F\) can hold a witness nearer than \(m\). For \(m=0\) the radius \(m-1\) falls outside the admissible range, and the theorem applies instead at \(r=0\), where every declaration is complete and \(\min\{0,1\}=0\); directly, \(s^0\in F\subseteq F^\star\) gives \(\mathrm{MSP}(S^\star,F^\star)=0\). The number of specifications evaluated is \(\lvert S\cap\bar B_m\rvert\); it is at most the ambient-product count \(\lvert\bar B_m\rvert\), because admissibility can only remove points. Counted axis by axis rather than through \(L_{\max}\), the enriched LaLonde declaration has \(1+2+2+1+3=9\) points at \(m=1\), against the \(13\) of the uniform display, while admissibility, which permits trimming only on the IPW arm, leaves six. The final claim is immediate: with \(F=\emptyset\) no radius contains a witness, so every radius up to \(K\) is enumerated.
\end{proof}

\begin{remark}[Two kinds of incompleteness]
\label{rem:two-incompleteness}
Completeness to radius \(r\) fails in two ways that are not equally verifiable. \emph{Option} incompleteness, a feasible combination of the declared axes' full option sets, within distance \(r\), that the declaration omits, is visible in the declaration itself and is removed by construction under \Cref{cor:anytime}. \emph{Axis} incompleteness, an analyst decision never declared as an axis, is not visible, and no enumeration can rule it out; calling it auditable would overstate what a reader can do with it. What \Cref{thm:certified} buys is that the residual becomes an explicit scope statement rather than a diffuse doubt: no undeclared decision loses support in fewer than the reported number of moves. That statement is substantive and must be defended substantively. \Cref{rem:axis-reach} gives the corresponding limit on the certificate's reach, and \Cref{cor:distance-one} is the case in which the statement carries no content at all.
\end{remark}

\begin{remark}[Why the hypothesis is exact rather than approximate]
\label{rem:no-approximate}
One might hope to weaken \(r\)-completeness to ``most specifications within radius \(r\) are declared'' and pay with a bound rather than an equality. No nontrivial bound is available: a single undeclared specification at distance one may be a witness, which forces \(\mathrm{MSP}(S^\star,F^\star)=1\) whatever the declaration reports. This is the necessity argument above with a single omitted point, and it is why the reliability of an incomplete declaration is a question about frequencies under realistic degradations rather than about worst-case bounds.
\end{remark}

\begin{remark}[Simultaneous bands do not satisfy (ii)]
\label{rem:sim-not-stable}
Condition (ii) asks that a specification's failure status not depend on the space it was declared in. This holds for the pointwise pair, whose intervals are computed specification by specification on a fixed resample matrix, and it fails for intervals read off a simultaneous band, whose critical value grows with \(\lvert S\rvert\). The two endpoints of \Cref{prop:sim-envelope} then behave differently under enlargement. For \(\widehat M^{\mathrm{sim}}_{\mathrm{lo}}\) both effects push the same way, a wider band lowers \(\ell(s)\) on every retained specification, and new specifications only add candidates, so it is weakly decreasing whenever the enlarged band is at least as wide, as it is for Bonferroni and for sup-\(t\). For \(\widehat M^{\mathrm{sim}}_{\mathrm{hi}}\) the effects oppose: a wider band raises \(u(s)\) and can remove members of \(\{u\le0\}\), while the added specifications can supply new ones, so no monotonicity holds. The certified interval is therefore a statement about the pointwise pair, and a simultaneous envelope must be recomputed on the space actually declared.
\end{remark}

The embedding hypotheses are not automatic. They hold for a full product declaration with the same baseline options, and for a richer space that keeps the declaration's admissibility rule, as the enriched LaLonde grid does (\Cref{app:empirics}); they can fail when the richer space changes a rule that couples axes, since a declared specification need not then remain admissible. The failure-set condition also fails whenever a specification present in both spaces is \emph{implemented} differently, since its interval and feasibility status may then change; \Cref{rem:sim-not-stable} records the corresponding issue for simultaneous bands. The two inclusions in \(F=F^\star\cap S\) do different work in \Cref{thm:certified}. The lower endpoint needs \(F^\star\cap S\subseteq F\): without it a specification could fail in the richer space while passing in the declared one, and no certificate would be possible. The upper endpoint needs \(F\subseteq F^\star\), and fails on its own without it, with \(S^\star=S\) and \(F^\star=\emptyset\) a declared witness ceases to be one, and the latent distance rises to \(\infty\) above any finite reported value.

% \begin{remark}[Embedding form]
% \label{rem:embedding}
% The hypotheses are often more convenient in mapped form. Call \((S^\star,F^\star)\) a \emph{feasibility-preserving enlargement} of \((S,F)\) if there is \(\iota:S\to S^\star\) with
% \[
% \iota(s^0)=s^{\star0},\qquad c^\star_0(\iota(s))=c_0(s)\ \ \forall s\in S,\qquad \iota(F)=F^\star\cap\iota(S).
% \]
% Both parts of the proof apply verbatim with \(\iota(s^\dagger)\) in place of \(s^\dagger\) and \(\iota(S)\) in place of \(S\), and \(r\)-completeness reads \(B_r\subseteq\iota(S)\). Injectivity, assumed in \Cref{thm:certified} so that \(S\) may be identified with \(\iota(S)\), is not needed in this form: both bounds use only that each \(t\in\iota(F)\) is \(\iota(u)\) for some \(u\in F\) of equal cost. The identity embedding recovers the statement as proved.
% \end{remark}

\paragraph{\Cref{prop:encoding} (Relabeling and re-encoding).}

\begin{proof}
(a) Let \(\pi(s)=(\pi_1(s_1),\dots,\pi_K(s_K))\). Each \(\pi_k\) is a bijection with \(\pi_k(s^0_k)=\tilde s^0_k\), so \(s_k\neq s^0_k\) if and only if \(\pi_k(s_k)\neq\tilde s^0_k\). Summing the indicators over \(k\) gives \(\tilde c_0(\pi(s))=c_0(s)\). Since estimates and intervals are carried along by \(\pi\), the feasible set is carried to \(\pi(F)\), and the two minima run over the same multiset of costs.

(b) Write \(\beta(s)=(\beta_1(s_1),\dots,\beta_K(s_K))\), so the code splits into blocks, one per axis, and Hamming distance is additive across blocks:
\[
c^\beta(s)=\sum_{k=1}^K d_H\bigl(\beta_k(s_k),\beta_k(s^0_k)\bigr).
\]
Fix \(k\). If \(s_k=s^0_k\) the \(k\)th term is \(0\). If \(s_k\neq s^0_k\) then \(\beta_k(s_k)\neq\beta_k(s^0_k)\) because \(\beta_k\) is injective, so the two blocks differ in at least one bit and the \(k\)th term is at least \(1\). Each term therefore dominates \(\mathbf 1\{s_k\neq s^0_k\}\), and summing gives \(c^\beta(s)\ge c_0(s)\) for every \(s\in S\). Taking minima over \(F\) preserves the inequality, so \(\mathrm{MSP}^\beta\ge\mathrm{MSP}\).

For strictness, take one axis with \(L_1=\{\ell_0,\ell_1,\ell_2\}\), baseline \(s^0_1=\ell_1\), and the one-hot code \(\beta_1(\ell_0)=(0,0)\), \(\beta_1(\ell_1)=(1,0)\), \(\beta_1(\ell_2)=(0,1)\). The move \(\ell_1\to\ell_2\) has \(c_0=1\) but \(d_H((1,0),(0,1))=2\). If the cheapest feasible specification is reached by exactly that move, \(\mathrm{MSP}=1<2=\mathrm{MSP}^\beta\). More generally, under a one-hot code with the reference option dropped, the \(k\)th block preserves cost for every move precisely when \(m_k=2\) or \(s^0_k\) is the reference option, since inflation requires a move between two distinct non-reference options and hence \(m_k\ge3\). Injectivity alone supports no such characterisation: the two-option code \(\beta_k(s^0_k)=(0,0)\), \(\beta_k(s_k)=(1,1)\) is injective and inflates every move even though \(m_k=2\).
\end{proof}

\begin{remark}[The axis-by-axis hypothesis is needed]
\label{rem:global-code}
Injectivity of a code chosen jointly across axes does not suffice. Take \(K=2\), \(L_1=L_2=\{0,1\}\), \(s^0=(0,0)\), and the injective code \((0,0)\mapsto00\), \((1,0)\mapsto01\), \((0,1)\mapsto11\), \((1,1)\mapsto10\). Then \(c_0((1,1))=2\) while \(c^\beta((1,1))=d_H(10,00)=1\), so \(c^\beta\ge c_0\) fails. What the proof of (b) uses is not injectivity alone but that the code decomposes into per-axis blocks, making Hamming distance additive over axes.
\end{remark}

\begin{remark}[Reading \(\mathrm{MSP}^\beta\) on the code image]
\label{rem:code-image}
\(c^\beta\) and \(\mathrm{MSP}^\beta\) are defined on \(\beta(S)\), not on the ambient binary cube. Under a one-hot code the cube contains strings activating two options of the same axis at once, which correspond to no specification; minimising over the cube would therefore not be a statement about the declared space at all.
\end{remark}

\paragraph{\Cref{prop:dir-validity} (Directional finite-sample validity).}
\emph{Statement.} Fix \(\alpha\in(0,1)\) and \(d=+1\). Let \(\mathcal N^*_{\mathrm{dir}}(S)=\{s\in S:\tau(s)\le0\}\) and \(\mathrm{MSP}^*_{\mathrm{dir}}=\mathrm{MSP}(S,\mathcal N^*_{\mathrm{dir}}(S))\). If \(\Pr(L(s)\le\tau(s))\ge1-\alpha\) for each \(s\in\mathcal N^*_{\mathrm{dir}}(S)\), then \(\Pr(\widehat{\mathrm{MSP}}_{\mathrm{dir}}\le\mathrm{MSP}^*_{\mathrm{dir}})\ge1-\alpha\) at every sample size.

\begin{proof}
If \(\mathcal N^*_{\mathrm{dir}}(S)=\emptyset\) then \(\mathrm{MSP}^*_{\mathrm{dir}}=\infty\) and the claim is immediate. Otherwise let \(s^*\in\mathcal N^*_{\mathrm{dir}}(S)\) attain the minimum, so \(\tau(s^*)\le0\) and \(c_0(s^*)=\mathrm{MSP}^*_{\mathrm{dir}}\). Because \(d\) is fixed in advance and \(\tau(\cdot)\) is a population quantity, \(s^*\) is non-random; this is where prespecification of the direction is used, and it is what allows the coverage assumption to be invoked at a single fixed specification. On the event \(\{L(s^*)\le\tau(s^*)\}\), which has probability at least \(1-\alpha\),
\[
L(s^*)\ \le\ \tau(s^*)\ \le\ 0,
\]
so \(s^*\in\widehat{\mathcal N}_{\mathrm{dir}}(S)\) and therefore \(\widehat{\mathrm{MSP}}_{\mathrm{dir}}\le c_0(s^*)=\mathrm{MSP}^*_{\mathrm{dir}}\).
\end{proof}

\paragraph{\Cref{cor:radius-test} (Testing a declared robustness radius).}
\emph{Statement.} Fix \(\alpha\in(0,1)\), \(d=+1\), a prespecified cost \(c\), and a threshold \(m>0\), and write \(S_{<m}=\{s\in S:c(s)<m\}\). For any distribution satisfying the coverage hypothesis of \Cref{prop:dir-validity} and having \(\mathrm{MSP}^*_{\mathrm{dir}}<m\), the rule rejecting when \(\widehat{\mathrm{MSP}}_{\mathrm{dir}}\ge m\) has rejection probability at most \(\alpha\).

\begin{proof}
On the event \(\{\widehat{\mathrm{MSP}}_{\mathrm{dir}}\le\mathrm{MSP}^*_{\mathrm{dir}}\}\) of \Cref{prop:dir-validity} we have \(\widehat{\mathrm{MSP}}_{\mathrm{dir}}\le\mathrm{MSP}^*_{\mathrm{dir}}<m\), so the rule does not reject. Hence
\[
\Pr\bigl(\widehat{\mathrm{MSP}}_{\mathrm{dir}}\ge m\bigr)
\le\Pr\bigl(\widehat{\mathrm{MSP}}_{\mathrm{dir}}>\mathrm{MSP}^*_{\mathrm{dir}}\bigr)\le\alpha .
\]
Only the total order on costs is used, so integrality of \(c_0\) plays no role and the statement is unchanged for the weighted cost of the Generality remark below. It is a level rather than a size in two distinct senses. It is indexed by the distribution, because the coverage hypothesis of \Cref{prop:dir-validity} is a condition on the distribution at hand, which fixes \(\tau\) and hence \(\mathcal N^*_{\mathrm{dir}}(S)\), and so fixes which coverage statements are being assumed, and it therefore holds across the null only where that hypothesis does. And it is not attained except in the least favourable configuration, in which the population nearest-failure specification sits at \(\tau=0\) with exact coverage while every other specification of \(S_{<m}\) has \(L(s)>0\) with probability approaching one.
\end{proof}

The equivalence behind the alternative is that \(\mathrm{MSP}^*_{\mathrm{dir}}\ge m\) holds exactly when \(\tau(s)>0\) at every \(s\in S_{<m}\), so \(H_{0,m}\) is the union null \(\bigcup_{s\in S_{<m}}\{\tau(s)\le0\}\) and the rule, which fires only when every \(s\in S_{<m}\) has \(L(s)>0\), is its intersection--union test \citep{berger1982multiparameter}. The absence of a multiplicity correction is the standard property of that construction rather than anything particular to MSP: some component null is true under \(H_{0,m}\), and the rule rejects only if that component's own level-\(\alpha\) test does. Two readings follow. Support at radius \(m\) is a universal statement over \(S_{<m}\) while loss of support is existential, which is why one witness settles a finite reported distance whereas a radius constrains every nearer specification; and the correction that does not appear in the level reappears in the power, which falls as \(\lvert S_{<m}\rvert\) grows, so a richer declaration is harder to defend at a given radius even though it is the more informative object (\Cref{thm:certified}). Finally, rejection is monotone in \(m\), so \(\widehat{\mathrm{MSP}}_{\mathrm{dir}}=\infty\) rejects at every threshold; this is one event and not a family of separate findings, and inverting the family returns the single one-sided bound \([\widehat{\mathrm{MSP}}_{\mathrm{dir}},\infty]\) that \Cref{prop:dir-validity} already asserts.

\paragraph{\Cref{prop:sim-envelope} (Sharp endpoints of a simultaneous-confidence envelope).}
\emph{Statement.} Fix \(\alpha\in(0,1)\) and \(d=+1\), and let \(\ell,u:S\to\mathbb R\) satisfy \(\ell(s)\le u(s)\) for every \(s\in S\) and \(\Pr\{\ell(s)\le\tau(s)\le u(s)\ \forall s\in S\}\ge1-\alpha\). With
\(\widehat M^{\mathrm{sim}}_{\mathrm{lo}}=\mathrm{MSP}(S,\{s:\ell(s)\le0\})\) and
\(\widehat M^{\mathrm{sim}}_{\mathrm{hi}}=\mathrm{MSP}(S,\{s:u(s)\le0\})\),
\[
\Pr\bigl(\widehat M^{\mathrm{sim}}_{\mathrm{lo}}\le\mathrm{MSP}^*_{\mathrm{dir}}\le\widehat M^{\mathrm{sim}}_{\mathrm{hi}}\bigr)\ge1-\alpha,
\]
and for a realized band the two endpoints are the extreme values of \(\mathrm{MSP}^*_{\mathrm{dir}}\) over all \(\tau\) consistent with it.

\begin{proof}
Write \(E=\{\ell(s)\le\tau(s)\le u(s)\ \forall s\in S\}\), so \(\Pr(E)\ge1-\alpha\), and recall \(\mathcal N^*_{\mathrm{dir}}(S)=\{s:\tau(s)\le0\}\).

\emph{Coverage.} Work on \(E\). If \(\tau(s)\le0\) then \(\ell(s)\le\tau(s)\le0\), so
\(\mathcal N^*_{\mathrm{dir}}(S)\subseteq\{s:\ell(s)\le0\}\); minimizing \(c_0\) over a larger set can only decrease the minimum, giving \(\widehat M^{\mathrm{sim}}_{\mathrm{lo}}\le\mathrm{MSP}^*_{\mathrm{dir}}\). If \(u(s)\le0\) then \(\tau(s)\le u(s)\le0\), so \(\{s:u(s)\le0\}\subseteq\mathcal N^*_{\mathrm{dir}}(S)\) and hence \(\mathrm{MSP}^*_{\mathrm{dir}}\le\widehat M^{\mathrm{sim}}_{\mathrm{hi}}\). Both inequalities hold on the single event \(E\), so no union bound is needed. The conventions \(\min\emptyset=\infty\) make the display correct in the degenerate cases: if no \(s\) has \(\ell(s)\le0\) then all three quantities are \(\infty\), and if no \(s\) has \(u(s)\le0\) the upper inequality is vacuous.

\emph{Sharpness.} Fix a realized band and let \(\mathcal T=\{\tau:\ell(s)\le\tau(s)\le u(s)\ \forall s\}\); the coverage argument shows every \(\tau\in\mathcal T\) has \(\mathrm{MSP}^*_{\mathrm{dir}}(\tau)\in[\widehat M^{\mathrm{sim}}_{\mathrm{lo}},\widehat M^{\mathrm{sim}}_{\mathrm{hi}}]\). We exhibit members attaining each endpoint. Sharpness is against the band alone: \(\mathcal T\) is the full box and imposes nothing linking the values \(\tau(s)\) across specifications, whereas those values are functionals of a single distribution. A model constraining them jointly can narrow the range, and the envelope is not claimed to be sharp against such a model.

For the lower endpoint, assume \(\widehat M^{\mathrm{sim}}_{\mathrm{lo}}<\infty\) and let \(s_\star\) attain it, so \(\ell(s_\star)\le0\) and \(c_0(s_\star)=\widehat M^{\mathrm{sim}}_{\mathrm{lo}}\). Set \(\tau_\star(s_\star)=\min\{0,u(s_\star)\}\), which lies in \([\ell(s_\star),u(s_\star)]\) because \(\ell(s_\star)\le0\), and satisfies \(\tau_\star(s_\star)\le0\); set \(\tau_\star(s)=u(s)\) at every other \(s\). Any \(s\) with \(c_0(s)<\widehat M^{\mathrm{sim}}_{\mathrm{lo}}\) has \(\ell(s)>0\) by minimality, hence \(\tau_\star(s)=u(s)\ge\ell(s)>0\) and \(s\notin\mathcal N^*_{\mathrm{dir}}\); specifications of cost \(\ge c_0(s_\star)\) cannot lower the minimum below \(c_0(s_\star)\). Thus \(\tau_\star\in\mathcal T\) and \(\mathrm{MSP}^*_{\mathrm{dir}}(\tau_\star)=c_0(s_\star)=\widehat M^{\mathrm{sim}}_{\mathrm{lo}}\). If instead \(\widehat M^{\mathrm{sim}}_{\mathrm{lo}}=\infty\) then \(\{s:\ell(s)\le0\}=\emptyset\), hence \(\{s:u(s)\le0\}=\emptyset\) and \(\widehat M^{\mathrm{sim}}_{\mathrm{hi}}=\infty\) as well, and the \(\tau^\star\) of the next paragraph attains both endpoints at once.

For the upper endpoint, take \(\tau^\star=u\), which lies in \(\mathcal T\). Then \(\tau^\star(s)\le0\) iff \(u(s)\le0\), so \(\mathcal N^*_{\mathrm{dir}}(\tau^\star)=\{s:u(s)\le0\}\) exactly and \(\mathrm{MSP}^*_{\mathrm{dir}}(\tau^\star)=\widehat M^{\mathrm{sim}}_{\mathrm{hi}}\); this covers the case \(\widehat M^{\mathrm{sim}}_{\mathrm{hi}}=\infty\) as well. Non-strictness of the criterion \(u(s)\le0\) is what makes \(\tau^\star=u\) admissible; under a strict criterion the supremum need not be attained.

\emph{Generality.} Only two properties of \(c_0\) were used: that the two feasible sets are nested as shown, and that \(F\subseteq F'\) implies \(\min_{F'}\le\min_F\). Both hold for \(\min_{s\in F}c(s)\) with an arbitrary declared cost \(c\), so the statement is unchanged for \(c_w(s)=\sum_kw_k\mathbf 1\{s_k\neq s^0_k\}\), \(w_k>0\), or any other prespecified cost.
\end{proof}

The lower half of the envelope is looser than \Cref{prop:dir-validity}. Whenever the band contains the pointwise intervals, as it does for Bonferroni and for sup-\(t\) at the same level, \(\ell(s)\le L(s)\) and so \(\{s:L(s)\le0\}\subseteq\{s:\ell(s)\le0\}\), giving \(\widehat M^{\mathrm{sim}}_{\mathrm{lo}}\le\widehat{\mathrm{MSP}}_{\mathrm{dir}}\): the band buys a weaker lower certificate under a stronger assumption. The content of \Cref{prop:sim-envelope} is therefore the upper endpoint, which \Cref{prop:dir-validity} cannot deliver at any price, since a single pointwise interval cannot certify that no cheaper specification has a non-positive estimand. Read as a test, that upper endpoint reverses \Cref{cor:radius-test}: on the envelope event, \(\widehat M^{\mathrm{sim}}_{\mathrm{hi}}<m\) is incompatible with \(\mathrm{MSP}^*_{\mathrm{dir}}\ge m\), so rejecting that null when \(\widehat M^{\mathrm{sim}}_{\mathrm{hi}}<m\) has error probability at most \(\alpha\); the argument is the containment already established and needs no separate proof.

\section{Supplementary Empirics}
\label[appendix]{app:empirics}

\paragraph{Literature-anchored grids.}
The additional two-source reanalysis uses only choices documented by the applied workflows.  For LaLonde, \citet{imbens2015matching} motivates the selected propensity-score terms, the CHIM overlap rule, three subclasses, 1:1 propensity-score matching, and the three post-adjustment sets.  We add an all-linear propensity-score comparator, retain all configurations in the relevant overlap facet, and use \(B=1{,}000\) stratified bootstrap resamples of the 185 treated and 260 randomized control units.  The no-trim and CHIM-trim analyses are separate facets because the latter excludes treated units; each facet has \(2\times2\times3=12\) configurations and the baseline is selected PS, three subclasses, and all-covariate outcome adjustment.

For Card--Krueger, the published robustness tables in \citet{card1993minimum} supply the one-at-a-time choices.  We form their full product within each fixed outcome/exposure facet: controls for chain and ownership; regional controls (admissible only in the wage-gap facet); temporary closures missing versus coded zero; retention versus exclusion of the NJ shore; absence versus presence of three wave-2 interview-date dummies; and retention versus exclusion of stores called back three or more times.  The published article does not give its date-bin cut points, so the three controls are implemented as four consecutive two-week bins between 5 November and 31 December 1992.  We facet rather than count FTE construction (manager inclusion and part-time weights \(0.4,0.5,0.6\)) and treatment contrast (wage gap versus NJ indicator), and use \(B=1{,}000\) store bootstrap resamples.  Hence there are 64 configurations in each wage-gap facet and 32 in each NJ-indicator facet, the latter excluding incompatible regional controls.  The new script writes the configuration-level and facet-level CSVs to \texttt{experiment/results\_literature\_anchored\_msp/}.

\paragraph{Witness-guided follow-up details.}
\label{app:witness-followup}
The routing protocol of \Cref{sec:followup} requires an axis type to be fixed
with the declaration. \Cref{tab:axis-types} records the three structures used
in our applications. The type is not inferred from the estimated intervals and
does not rank one option over another; it only determines which follow-up the
declared structure admits.

\begin{table}[H]
    \centering
    \begin{tabular}{@{}p{0.12\linewidth}p{0.28\linewidth}p{0.20\linewidth}p{0.32\linewidth}@{}}
        \toprule
        type & declared structure, with examples & what a witness names & follow-up operator \\
        \midrule
        ordered &
        totally ordered levels of one quantity: trimming fraction, bandwidth,
        caliper, winsorization point, penalty &
        the nearest level &
        search between baseline and witness for the first loss of support; with
        two ordered axes, trace a failure frontier \\[3pt]
        nested &
        options ordered by inclusion: covariate sets, control sets, sample
        restrictions, fixed effects &
        the nearest set change &
        decompose the set difference and report every minimal failing component
        change \\[3pt]
        categorical &
        options with no declared order: OLS versus IPW, matching versus
        subclassification, link function &
        the nearest alternative &
        give a structured contrast of the analyses; impose no continuous path
        where the declaration supplies none \\
        \bottomrule
    \end{tabular}
    \caption{Declared axis types and their witness-guided follow-ups. A
    categorical follow-up is partly a refusal: the protocol does not manufacture
    an ordering that the declaration lacks.}
    \label{tab:axis-types}
\end{table}

Facets are separated before routing because a choice that changes the estimand
is not an ordinary axis. For example, trimming is an axis in the enriched raw
LaLonde grid but a facet in the published-pipeline reanalysis, where it removes
treated units and changes the target population. The protocol then retains
\emph{all} nearest witnesses. For the relevant failure set \(F\), an axis is
load-bearing when every element of \(W_F\) changes it, an alternative route
when only some elements do, and locally irrelevant when none do. At the
H\'ajek baseline of the enriched declaration, ten witnesses attain distance
two and no axis is common to all ten; the correct report is therefore a
distributed failure rather than an arbitrarily selected decision.

Ordered and nested follow-ups can add analyses that were absent from the
original declaration. If the baseline supports the claim and
\(\mathrm{MSP}(S,F)=1\), preserving the baseline, costs, old analyses, and
their failure status gives the enlargement conditions of
\Cref{cor:distance-one}; the targeted distance stays one. This statement is
coordinate-specific. The entire pair is invariant only when both coordinates
are one, as at the primary baseline of the enriched LaLonde declaration. Away
from distance one the certified result is a bracket, so a refined analysis is
reported as a new declaration rather than as an updated MSP.

\emph{Categorical-axis implementation contrast.}
The published-pipeline witness changes three-subclass adjustment to
\(1{:}1\) propensity-score matching. Our executed matcher uses
\texttt{np.argmin}; when controls are exactly equally near, it therefore takes
the first occurrence. Holding the selected propensity model, all-covariate
outcome adjustment, sample, and common stratified resamples fixed,
\Cref{tab:matching-tie-followup} compares that implementation with a
deterministic, outcome-blind rule that divides a treated unit's weight equally
over all of its exactly nearest controls.

\begin{table}[H]
    \centering
    \begin{tabular}{@{}lrrrrr@{}}
        \toprule
        estimator implementation & \(\hat\tau\) & 95\% CI &
        positive-weight controls & control ESS & top-five weight \\
        \midrule
        three subclasses & 1610 & \([459,3201]\) & 260 & 206.4 & 3.5\% \\
        matching, first exact tie & 374 & \([-1121,4138]\) & 59 & 17.4 & 46.0\% \\
        matching, average exact ties & 1785 & \([157,3230]\) & 209 & 120.6 & 11.9\% \\
        \bottomrule
    \end{tabular}
    \caption{Structured follow-up of the categorical matching witness
    (\(B=1{,}000\), the same bootstrap draws as the literature-anchored grid).
    The original matching implementation reaches null contact; averaging over
    exact ties restores directional support and sharply reduces control-weight
    concentration. ESS denotes effective sample size.}
    \label{tab:matching-tie-followup}
\end{table}

The contrast exposes substantial hidden multiplicity. Among 185 treated units,
108 have multiple exactly nearest controls; conditional on a tie, the median
number is 21 and the maximum is 47. Resolving those ties uniformly at random
over 10,000 outcome-blind draws gives point estimates with median \(1788\) and
2.5--97.5 percentiles \([1317,2241]\). This is an algorithmic sensitivity
distribution, not a confidence interval. Extending the same bootstrap stream to
\(B=4{,}000\) confirms the classification: subclassification gives
\([350,3124]\), first-tie matching gives \([-1107,4060]\), and
average-tie matching gives \([86,3228]\). The follow-up therefore does not
show that matching is inferior; it shows that the coarse estimator label hides
a consequential implementation decision. The original distance is not
recomputed because its first-tie witness remains in the declaration. Exact
outputs are written by
\texttt{block15\_matching\_tie\_followup.py}.

\emph{Nested-axis component analysis.}
The primary raw-earnings null-contact witness changes the covariate set from
\texttt{basic} to \texttt{full}. Holding the estimator, functional form,
trimming, and outcome scale fixed, we enumerate all 16 subsets of the four
added covariates on the same \(B=4{,}000\) resample matrix and seed as the
reported grid. The empty addition reproduces the baseline and the complete
addition reproduces the witness to the last reported digit, checking that the
two original analyses remain embedded unchanged.

\begin{table}[H]
    \centering
    \setlength{\tabcolsep}{3pt}
    \begin{tabular}{@{}lrcl@{}}
        \toprule
        addition to \texttt{basic} & \(\hat\tau\) & 95\% CI & status \\
        \midrule
        none (baseline) & \(-3807\) & \([-4933,-2654]\) & supports \\
        \texttt{married} & \(-2554\) & \([-3680,-1370]\) & supports \\
        \texttt{nodegree} & \(-3473\) & \([-4603,-2310]\) & supports \\
        \texttt{re74} & \(269\) & \([-911,1504]\) & null contact \\
        \texttt{re75} & \(742\) & \([-411,1922]\) & null contact \\
        \texttt{married}+\texttt{nodegree} & \(-2331\) & \([-3465,-1134]\) & supports \\
        all four (witness) & \(879\) & \([-280,2077]\) & null contact \\
        \bottomrule
    \end{tabular}
    \caption{Component analysis of the distance-one LaLonde witness. The nine
    additions not displayed all reach null contact: support is retained exactly
    on subsets of \(\{\texttt{married},\texttt{nodegree}\}\). No addition
    reverses the claim.}
    \label{tab:component-analysis}
\end{table}

There are two minimal failing additions, \texttt{re74} and \texttt{re75};
adding \texttt{married} or \texttt{nodegree} alone preserves support. Thus,
within this nested axis, loss of support is attributable to pre-treatment
earnings and not to the other additions. This reproduces a distinction already
familiar in the LaLonde literature and is therefore a check on the routing
protocol rather than a new substantive finding. The decomposition identifies
where an argument is required; it does not establish which adjustment set is
causally preferable.

\paragraph{Outcome-scale check for the NSW randomized sample.}
\label{par:nsw-scale-check}
In the NSW randomized sample, the randomization-reference grid includes an outcome-scale axis (raw earnings vs.\ \(\log(1+\text{earnings})\)), so pooled configurations mix estimands. To confirm that the observed \(\mathrm{MSP}=\infty\) is not an artefact of cross-scale cancellation, we recompute MSP separately on the two scale-specific \(K=3\) subgrids (8 configurations each), holding fixed the same bootstrap procedure and resamples. Both scale-specific grids also yield \(\mathrm{MSP}=\infty\): no configuration has a 95\% CI that contains zero.

\paragraph{MSP vs.\ Fragility Index: construction details.}
\label{par:fragility}
The standard Fragility Index \citep{walsh2014statistical} is defined for 2\(\times\)2 tables and Fisher's exact test. To enable a direct comparison on the NSW continuous-outcome data, we define \(\widetilde{\mathrm{FI}}\) as the minimum number of treated-unit post-period earnings that must be set to zero for the bootstrap 95\% CI to contain zero (baseline specification: OLS, basic covariates, linear model, no trimming; \(B=500\) fixed resamples). We report two variants: \(\widetilde{\mathrm{FI}}_{\mathrm{adv}}\) zeros the highest-earning treated units first (adversarial ordering); \(\widetilde{\mathrm{FI}}_{\mathrm{rnd}}\) reports the median over 50 random orderings. Unlike MSP, which perturbs analyst choices while holding the data fixed, \(\widetilde{\mathrm{FI}}\) perturbs the data while holding the specification fixed. In the adversarial case, the three earnings zeroed are \$60{,}308, \$36{,}647, and \$34{,}099. Results and interpretation are reported in the main text (\Cref{tab:fragility-vs-msp}).

\paragraph{Admissibility on the LaLonde raw grid.}
\label{par:lalonde-admissible}
The trimming axis of the raw-earnings grid selects on the estimated propensity score, which the OLS specifications never compute. Pairing an OLS specification with either trimming option therefore yields the same analysis, and the two cells agree to the cent in point estimate and interval alike. We report the declared space as the admissible subset, \(\lvert S\rvert=12\): four configurations on the OLS arm, where trimming sits at its baseline option, and eight on the IPW arm, where both trimming options are meaningful. This is the concrete instance of the admissibility constraints that make \(S\) a proper subset of \(\prod_k L_k\) in \Cref{sec:framework}.

Every reported minimum is unchanged by this restriction. Each inadmissible cell duplicates an admissible one whose distance from the baseline is \emph{smaller} by exactly one, since the two differ only in the trimming coordinate; a set of candidates cannot change its minimum when an element is removed that was tied in interval but strictly costlier. Hence \((\mathrm{MSP},\mathrm{MSP}_{\mathrm{rev}})=(1,3)\), the \(\alpha\)-sensitivity check, and the CI-construction ablation are the same on the product grid and on the admissible subset. What the restriction does change is any statistic that counts configurations: the fraction of intervals containing zero is \(3/12\) rather than \(5/16\), and counts by distance are no longer weighted twice on the OLS arm.

\paragraph{The enriched LaLonde declaration.}
\label{par:lalonde-enriched}
The enriched grid reported in \Cref{sec:lalonde} keeps the four axes, the baseline, and the admissibility rule of the binary grid, and widens each option set: estimator (OLS, H\'ajek IPW, Horvitz--Thompson IPW), covariates (basic, full, full with pairwise products of the continuous covariates), functional form (linear, quadratic), and trimming at \(0\%\), \(1\%\), \(5\%\), \(10\%\). Trimming remains admissible only on the IPW arm, giving \(\lvert S\rvert=54\): six configurations on the OLS arm and forty-eight on the IPW arm. Intervals are percentile bootstrap at \(B=4{,}000\) on a common resample matrix, with the same seed as the binary grid. The H\'ajek option reproduces the binary grid's IPW estimates exactly, so the binary grid embeds in the enriched one and the hypotheses of \Cref{thm:certified} hold by construction rather than by assumption: the baseline row of the enriched grid reproduces the binary grid's estimate and interval to every reported digit (\(-3807.207\), \([-4933.205,-2654.450]\)), which is the empirical form of condition (ii). The binary grid is \(0\)-complete in the enriched one but not \(1\)-complete, the enriched declaration carries five admissible specifications at distance one against the binary grid's three, the two additions being the third estimator option and the interacted adjustment set, so \Cref{thm:certified} certifies the binary grid's \(\mathrm{MSP}_{\mathrm{null}}=1\) outright and brackets its \(\mathrm{MSP}_{\mathrm{rev}}=3\) only within \([1,3]\). The enriched value of \(1\) sits at the bottom of that bracket.

Twenty-two of the 54 intervals contain zero and five lie strictly on the positive side. The nearest reversal is at distance one; the four at distance three include the binary grid's witness, H\'ajek IPW with full covariates and \(10\%\) trimming, at \(\hat\tau=\$2{,}036\) and CI \([\$278,\$3{,}396]\).

\emph{A withdrawn option.} A cubic term was declared alongside the quadratic and then withdrawn before any value was read from the enriched grid. Cubing unscaled earnings puts design entries near \(6\times10^{13}\); the propensity fit degenerates and unnormalised Horvitz--Thompson weights explode, returning effects of \$3.6M and \$5.1M on a programme whose outcome averages a few thousand dollars. Such an analysis is not a defensible declared option, and dominated options should be excluded from a declaration rather than reported. We record the withdrawal because the grounds matter: it was made on numerical grounds, not on the value of MSP, and the pair is \((1,1)\) whether the cubic option is in the space or out of it. One configuration in the retained space, Horvitz--Thompson with interactions and no trimming, still returns \$139{,}940; the instability is a property of unnormalised weights under a poorly conditioned propensity model, and every trimmed Horvitz--Thompson configuration is well behaved. That configuration is not a witness for either component.

\paragraph{Exact declaration-space stress test.}
\label{par:declaration-stress}
The conditionality of MSP can be stress-tested without a new model fit or a
new bootstrap. Starting from the enriched 54-configuration reference
declaration \(S^\star\), we enumerate every product-form sub-declaration that
retains its baseline option on each axis. There are
\(2^2\times2^2\times2^1\times2^3=256\) such declarations per baseline:
estimator has two non-baseline options, covariates two, functional form one,
and trimming three. The IPW-only trimming admissibility rule is held fixed,
so every resulting feasible set is a genuine subspace of \(S^\star\). The
enumeration is exact; the displayed shares below are conditional on the
indicated structural stratum rather than probabilities under an arbitrary
uniform distribution over declarations.

\begin{table*}[t]
    \centering
    \begin{tabular}{@{}lrrrrr@{}}
        \toprule
        reference baseline & \(\lvert S\rvert\) & \(n\) & exact recovery & overstate \(\ge1\) & false \((\infty,\infty)\) \\
        \midrule
        OLS (either functional form) & 4  & 28 & 71.4\% & 28.6\% & 21.4\% \\
        OLS (either functional form) & 8  & 22 & 72.7\% & 27.3\% & 13.6\% \\
        OLS (either functional form) & 12 & 22 & 100.0\% & 0.0\% & 0.0\% \\
        OLS (either functional form) & 24 & 6  & 100.0\% & 0.0\% & 0.0\% \\
        OLS (either functional form) & 54 & 1  & 100.0\% & 0.0\% & 0.0\% \\
        \midrule
        H\'ajek IPW & 4  & 24 & 41.7\% & 58.3\% & 58.3\% \\
        H\'ajek IPW & 8  & 32 & 78.1\% & 21.9\% & 18.8\% \\
        H\'ajek IPW & 12 & 35 & 97.1\% & 2.9\% & 0.0\% \\
        H\'ajek IPW & 24 & 14 & 100.0\% & 0.0\% & 0.0\% \\
        H\'ajek IPW & 54 & 1  & 100.0\% & 0.0\% & 0.0\% \\
        \midrule
        HT IPW & 4  & 24 & 50.0\% & 50.0\% & 33.3\% \\
        HT IPW & 8  & 32 & 68.8\% & 31.3\% & 6.3\% \\
        HT IPW & 12 & 35 & 82.9\% & 17.1\% & 0.0\% \\
        HT IPW & 24 & 14 & 92.9\% & 7.1\% & 0.0\% \\
        HT IPW & 54 & 1  & 100.0\% & 0.0\% & 0.0\% \\
        \bottomrule
    \end{tabular}
    \caption{Exact declaration-space stress test for enriched LaLonde (\(256\) product-form declarations per baseline). Shares by \(\lvert S\rvert\) report recovery, overstatement, and false universal support.}
    \label{tab:declaration-stress}
\end{table*}

The stress test makes the direction of refinement visible without treating it
as a discovery: every pointwise gap is nonnegative by construction. One-option
drops and trimming-only coarsenings recover the reference directional distance
for every baseline in this grid, whereas more aggressive mixed omissions can
hide every nearest witness and produce a false \((\infty,\infty)\) report. The
nearest-witness identity above explains which options carry the result. For
the OLS baseline, the null witness is the full-covariate option and the
reversal witness is the full-interaction option; the directional target is
lost when both are omitted. For HT, the two nearest directional witnesses are
the 10\% trimming option and the full-interaction option. H\'ajek has a more
distributed witness family: hiding every directional witness requires omitting
combinations of three or five options rather than one load-bearing option.

The same fixed bootstrap matrix permits a separate simultaneous stress test:
recomputing the sup-\(t\) critical value for each declaration gives
\(q_S=4.21184\) on the full 54-configuration reference declaration, where the
simultaneous envelope is \([1,\infty]\) for all four baselines, and across
all nested declaration pairs the possible-loss endpoint
\(\widehat M_{\mathrm{lo}}\) is monotone in the refinement direction.

\paragraph{Weighted costs.}
\label{par:weight-template}
The general arbitrary-cost formulation is already covered by the generality clause in the proof of \Cref{prop:sim-envelope}: one may replace \(c_0\) by a declared cost \(c_w(s)=\sum_k w_k\mathbf{1}\{s_k\neq s_k^0\}\), with \(w_k>0\). \Cref{thm:certified} is the exception: its lower endpoint \(r+1\) counts one further decision and uses the integrality of \(c_0\), so under a weighted cost it would have to be replaced by the smallest cost attainable in \(S^\star\) exceeding \(r\), a quantity the declaration alone does not determine. The practical issue is not the algebra but the justification of the scale. Our recommendation is to use one unit per analyst decision by default; to use ordinal distances only when the spacing of levels has substantive meaning; and to discretize continuous hyperparameters on a declared, policy-relevant grid or tolerance. Weights should be fixed from external considerations such as analyst time, institutional cost, or substantive importance, rather than fitted to the observed outcome. A weighted result should report the normalization, \(w_{\max}/w_{\min}\), and sensitivity over a small prespecified set of plausible vectors.

For transparency, we also recompute the enriched LaLonde pair under two post hoc sensitivity vectors that make the covariate-adjustment decision costlier, holding the other three axis weights at one. This is not a prespecified or outcome-independent weighting exercise and does not identify a preferred scale; it shows the dependence that a weighted report must disclose. Under unit costs, the nearest null-contact and reversal witnesses are the full and interacted covariate adjustments, respectively, each at cost one. At \(w=(1,2,1,1)\), both distances become two; at \(w=(1,4,1,1)\), the nearest null-contact witness becomes the two-axis HT-plus-10\%-trimming specification (cost two), while the interacted adjustment remains the reversal witness at cost four. Thus, a Hamming-distance-one finding remains a fact about the unit-cost declaration, but neither its numerical value nor its weighted minimizer is invariant to reweighting. The exact rows are written by \texttt{block13\_declaration\_stress\_test.py} to \texttt{weighted\_cost\_sensitivity.csv}.

\paragraph{A prespecification template.}
\Cref{tab:prespec-template} lists the fields a declaration should fix before analysis.

\begin{table}[h]
    \centering
    \begin{tabular}{@{}p{0.12\linewidth}p{0.10\linewidth}p{0.11\linewidth}p{0.19\linewidth}p{0.19\linewidth}p{0.07\linewidth}@{}}
        \toprule
        axis & type & base & levels & rationale / estimand & wt. \\
        \midrule
        estimator & categorical & OLS & OLS, H\'ajek, HT & same target; fixed scale & \(w_1\) \\
        covariates & nested & basic & basic, full, interactions & pre-treatment adjustment & \(w_2\) \\
        functional form & categorical & linear & linear, quadratic & same outcome/target & \(w_3\) \\
        trimming & ordered & none & none, 1\%, 5\%, 10\% & overlap rule; IPW only & \(w_4\) \\
        \bottomrule
    \end{tabular}
    \caption{Illustrative prespecification template. The weight column records the declared cost of changing each axis, one by default. The type column records the structure the axis's options are claimed to carry and selects the follow-up operator of \Cref{sec:followup}; it does not enter the distance. A completed declaration should also record the rationale for each level, excluded combinations, estimand-preservation checks, CI procedure, and all random seeds. The entries shown here describe the enriched LaLonde declaration; they are not a claim that the four weights are equal in every application.}
    \label{tab:prespec-template}
\end{table}

\paragraph{Joint application with hidden-bias sensitivity measures.}
\label{par:hidden-bias-joint}
We apply MSP and representative hidden-bias summaries to the same LaLonde data. The two rows below use the same negative baseline but different landscapes: \(S_1\), of four admissible configurations, varies covariates and functional form while holding the estimator at OLS, on which trimming is inadmissible; \(S_2\), of six, varies estimator, functional form, and trimming, the last on the IPW arm only, while holding covariates at basic. E-values and Cinelli--Hazlett robustness values are functions of the fixed baseline regression, so they are identical across the two landscapes by construction, whereas MSP changes with the documented analyst-decision space.

\begin{table}[h]
    \centering
    \begin{tabular}{@{}p{0.31\linewidth}p{0.18\linewidth}ccc@{}}
        \toprule
        landscape & MSP pair & E-value (point) & E-value (CI) & RV\(_{1,\alpha}\) \\
        \midrule
        \(S_1\): covariates, form & \((1,\infty)\) & 2.212 & 1.813 & 0.026 \\
        \(S_2\): estimator, form, trim & \((\infty,\infty)\) & 2.212 & 1.813 & 0.026 \\
        \bottomrule
    \end{tabular}
    \caption{MSP distinguishes declared landscapes sharing a baseline, whereas hidden-bias summaries do not.}
    \label{tab:hidden-bias-joint}
\end{table}

This is a complementarity result, not a ranking: the hidden-bias measures quantify how strong an unmeasured-confounding relationship would need to be for a fixed regression to change, while MSP quantifies how close a declared analyst decision is to loss of support. Neither quantity validates the causal claim by itself.

\paragraph{Synthetic FI--MSP flip experiment: construction and results.}
\label{par:fi-msp-flip}
The NSW randomized sample establishes one direction of divergence between data fragility and specification fragility. To show that the contrast is bidirectional, we simulate two regimes using the same \(K=4\) binary MSP axes as in Block 3: omit \(X_1\), omit \(X_2\), force linearity, and disable trimming. MSP is computed from common-bootstrap percentile CIs over all 16 configurations. Within each replicate, FI is computed on the baseline specification only, using a separate fixed bootstrap index matrix and an adversarial ordering that zeros treated outcomes from highest \(Y\) downward.

\emph{Regime A (tail-driven, no null contact).} Treatment assignment is nearly random, with weak dependence on \(X_1\) and \(X_2\), so analyst choices induce little bias variation across specifications. A small random subset of treated units receives a large outcome bonus that is orthogonal to covariates. This concentrates evidential support in a few treated observations: zeroing those outcomes can quickly make the baseline CI contain zero, but changing specification rarely does.

\emph{Regime B (diffuse-effect, specification-sensitive).} The treatment effect is homogeneous across treated units, but treatment assignment depends strongly on \(X_1\), \(X_2\), and \(X_3\), while untreated outcomes depend negatively on \(X_1\) and \(X_2\). Under this opposite-sign confounding structure, the baseline specification recovers the effect, but omitting \(X_1\) is often enough to produce a null-compatible CI. Because the effect is diffuse rather than concentrated in a few treated observations, many adversarial zeroings are required before the baseline CI contains zero.

The paper-quality run uses \(n=800\), \(B_{\mathrm{MSP}}=100\), \(B_{\mathrm{FI}}=100\), and 120 Monte Carlo replicates per regime. In Regime A, the tail bonus is assigned to 3\% of treated units with magnitude \(15\sigma\); this yields a baseline significant result in 97.5\% of replicates while keeping the signal concentrated. In Regime B, the diffuse treatment effect is \(\tau=0.7\), strong enough for the baseline to be significant in every replicate but still vulnerable to the single-axis omission of \(X_1\).

\begin{table}[h]
    \centering
    \begin{tabular*}{\linewidth}{@{\extracolsep{\fill}}lcc@{}}
        \toprule
        summary & \shortstack{Regime A:\\tail-driven} & \shortstack{Regime B:\\spec-sensitive} \\
        \midrule
        \(P(\)baseline significant\()\) & 0.975 & 1.000 \\
        median FI & 7 & 38 \\
        mean FI & 7.24 & 38.19 \\
        \(P(\mathrm{FI}\le 3)\) & 0.067 & 0.000 \\
        \(P(\mathrm{FI}\le 10)\) & 0.883 & 0.000 \\
        median FI / \(n_{\mathrm{treated}}\) & 1.75\% & 10.27\% \\
        \(P(\mathrm{MSP}_{\mathrm{null}}=1)\) & 0.017 & 0.925 \\
        \(P(\mathrm{MSP}_{\mathrm{null}}=\infty)\) & 0.958 & 0.033 \\
        median finite \(\mathrm{MSP}_{\mathrm{null}}\) & 0 & 1 \\
        \bottomrule
    \end{tabular*}
    \caption{FI and null-contact MSP reverse their ordering across a data-fragile regime without null contact and a data-robust, specification-fragile regime (\(n=800\), 120 replicates each). Neither measure substitutes for the other.}
    \label{tab:fi-msp-flip-appendix}
\end{table}

\Cref{tab:fi-msp-flip-appendix} gives the detailed synthetic evidence summarized in the main text, adding mean FI, the tail probabilities \(P(\mathrm{FI}\le 3)\) and \(P(\mathrm{FI}\le 10)\), the treated-fraction normalization, and the median finite \(\mathrm{MSP}_{\mathrm{null}}\). In Regime A the conclusion is data-fragile, yet no declared specification reaches null contact in 95.8\% of replicates; the reversal side was not recorded in this run, so we do not read that as universal directional support. In Regime B the conclusion is data-robust yet specification-fragile. The two metrics are therefore best read jointly.

\paragraph{Noise-probe calibration of the raw-earnings pair.}
\label{par:noise-probes}
The calibration of \Cref{sec:lalonde} attaches an internal reference scale to the pair \((1,3)\) using \emph{noise-augmentation probes}: synthetic axes constructed to be irrelevant in expectation. It is post hoc and descriptive, it was designed after the pair was known, its protocol (probe construction, seeds, ensemble sizes, and reported statistics) was frozen before execution, and it attaches no error rate. Probes come in two types that are never pooled, because the propensity model is an unstandardized penalized logistic and therefore scale-sensitive. A \emph{block probe} is a joint row-permutation of the whole full-minus-basic covariate block \(\{\textit{married},\textit{nodegree},\textit{re74},\textit{re75}\}\); it preserves each marginal and the within-block dependence, severs the association with treatment, outcome, and the retained covariates in expectation, and, critically, adds four columns, the same size of operation as the distance-one witness flip. A \emph{Gaussian probe} adds a single \(N(0,1)\) covariate and calibrates the smaller operation of adding one null column. Two designs were considered and excluded: a bootstrap-seed axis would break the common resample matrix on which every interval in the grid is computed, and a subsample-jitter axis changes the target population and is therefore not estimand-preserving. Probe vectors are generated once from a declared seed; the run records the PRNG, the seed, and a SHA-256 hash of every probe matrix.

For future applications, the probe family and its null-in-expectation rationale, the operation size, the ensemble and seeds, the summary statistic, and the non-pooling rules should all be fixed before inspecting the MSP pair; the present run satisfies freeze-before-execution but was designed after the original pair was known, and so retains its post hoc status.

Fifty probes of each type are evaluated as single flips from the baseline, on the common resample matrix at \(B=4{,}000\); the reported statistics are empirical quantiles (linear interpolation) of the absolute shift in the directional support bound \(b_{-}(s)=-U(s)\). With \(M=50\) the upper quantiles are order statistics of the ensemble and carry no probability semantics. The distance-one witness flip (basic to full covariates) shifts the bound by \$4{,}732; block probes shift it by at most \$31 (q95 \$22, median \$8), and Gaussian probes by at most \$15 (q95 \$12). The witness response is a factor of 152 beyond the largest block-probe response; Gaussian results remain a separate calibration of adding one null column. \Cref{fig:noise-probes} plots both ensembles against the witness on a log scale, where the two orders of magnitude separating them are visible directly. The reversal witness at distance three shifts the bound by \$6{,}051 as a whole configuration.

Two block probes, declared as axes \(P_1,P_2\), augment the admissible grid to \(\lvert S^{+}\rvert=12\times4=48\). The augmented pair is \((1,3)\), no minimal witness of either component contains a probe coordinate, and the probe-only restriction \(\{(s^0,p_1,p_2)\}\) has pair \((\infty,\infty)\): no combination of declared noise loses support. Equality between the original and augmented pairs is therefore a consistency check, not evidence. A simulation validation on the data-generating processes of the main simulation (\Cref{sec:simulation}; \(R=200\) replicates per regime, \(n=800\), \(B=200\), claim direction \(+1\), ten permuted-block and ten Gaussian probes per replicate, common resamples within a replicate) checks that the probe scale separates known-active operators from noise without conditioning on minimum witnesses. Candidate operators are first validated at the population level (\(n=2\times10^{5}\)): omitting either confounder is active in both regimes, forcing linearity is active only in the interaction regime (the flexible model retains the \(X_1X_2\) term, so axis names are not trusted), and disabling trimming is inert in both. Across replicates the known-active operators shift the bound by a median of 0.81 (additive regime) and 0.60 (interaction), while probe shifts have q95 of at most 0.015 in every regime--type cell, a separation of more than fortyfold, and the entire content of the claim.

\begin{figure}[t]
    \centering
    \includegraphics[width=0.5\linewidth]{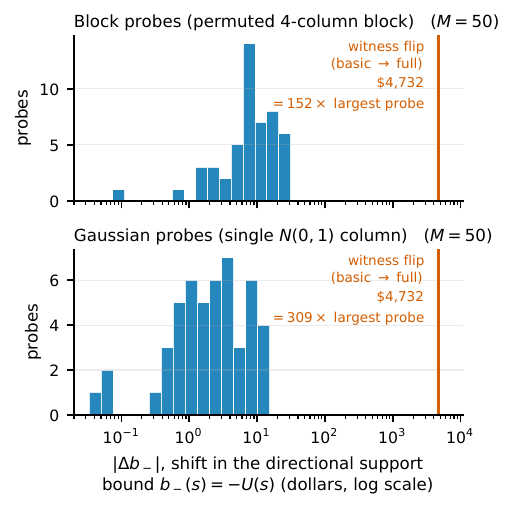}
    \caption{The covariate-adjustment decision moves the LaLonde support bound by more than two orders of magnitude beyond anything declared noise achieves. Each histogram shows absolute shifts in the directional support bound \(b_{-}(s)=-U(s)\) for fifty probes evaluated as single flips from the baseline (common resamples, \(B=4{,}000\), log scale); the vertical line is the distance-one witness flip from basic to full covariates. Block probes match the witness in operation size, four columns, and still top out at \$31, a factor of 152 below it; single Gaussian probes calibrate the smaller one-column operation separately, and the two types are never pooled. The separation is what places the witness outside the declared noise ensemble.}
    \label{fig:noise-probes}
\end{figure}

\paragraph{Simulation data-generating processes.}
\label{par:dgp}

\emph{Main simulation (\Cref{tab:directional-summary,tab:decision-rules}).}
Four covariates \(X_1,X_2,X_3,X_4\overset{\mathrm{iid}}{\sim}\mathcal{N}(0,1)\) are drawn independently. Treatment is assigned via a logistic model, \(A\mid X\sim\mathrm{Bernoulli}(\sigma(\eta))\) with \(\sigma(t)=(1+e^{-t})^{-1}\) clipped to \([0.02,0.98]\). Under the additive regime:
\[
\eta = -0.2 + 0.8X_1 + 0.8X_2 + 1.1X_3,
\qquad
\]
\[
g(X) = -1.2X_1 - 1.0X_2 - 0.6X_3 + 0.2X_4.
\]
Under the interaction-heavy regime:
\[
\eta = -0.2 + 0.8X_1 + 0.6X_2 + 1.1X_3 + 1.0X_1X_2,
\]
\[
\begin{aligned}
g(X)={}&-1.1X_1 - 0.7X_2 - 0.6X_3 + 0.2X_4\\
&- 1.0X_1X_2 - 0.8X_2^2.
\end{aligned}
\]
The outcome is \(Y = \tau A + g(X) + \varepsilon\) with \(\varepsilon\sim\mathcal{N}(0,1)\). Opposite-sign confounding (treated units selected on covariates associated with lower untreated outcomes) means naive adjustment attenuates \(\hat\tau\). Within each replicate, dispersion is the standard deviation of the 16 point estimates across configurations; \Cref{tab:directional-summary} reports its average over replicates.

\begin{table}[h]
    \centering
    % \begin{tabular}{@{}rcccc@{}}
        \begin{tabular*}{\linewidth}{@{\extracolsep{\fill}}lcccc@{}}

        \toprule
        \(\tau\) & \(P(\mathrm{MSP}_{\mathrm{null}}=\infty)\) &
        \(P(\mathrm{MSP}_{\mathrm{rev}}<\infty)\) &
        med. \(\mathrm{MSP}_{\mathrm{dir}}\) & dispersion \\
        \midrule
        0.0 & 0.085 & 1.000 & 0 & 0.498 \\
        0.3 & \textbf{0.815} & \textbf{1.000} & 1 & 0.500 \\
        0.7 & 0.025 & 1.000 & 1 & 0.499 \\
        1.5 & 0.080 & 0.010 & 2 & 0.503 \\
        \bottomrule
    \end{tabular*}
    \caption{Directional simulation (\(n=800\), \(B=100\), 200 replicates per \(\tau\)): the pair distinguishes overshoot from genuine universal directional support.}
    \label{tab:directional-summary}
\end{table}

\emph{FI--MSP flip experiment (Appendix, \Cref{tab:fi-msp-flip-appendix}).}
Regime A (tail-driven): treatment is nearly random, \(\eta=0.05X_1+0.05X_2\), and \(g(X)=0.1X_1+0.1X_2\). Significance is driven by a tail bonus assigned to a proportion \(p_{\mathrm{tail}}=0.03\) of treated units, each receiving an additive boost of \(15\sigma\) to their outcome; no specification omission removes this signal so \(\mathrm{MSP}_{\mathrm{null}}=\infty\) in almost every replicate. Regime B (spec-sensitive): \(\eta=-0.2+1.2X_1+0.8X_2+0.6X_3\) and \(g(X)=-1.5X_1-1.0X_2-0.5X_1X_2\), with \(\tau=0.7\). The strong opposite-sign confounding via \(X_1\) means that omitting \(X_1\) alone is sufficient to flip the CI, so \(\mathrm{MSP}_{\mathrm{null}}=1\) in almost every replicate.

\paragraph{Replicate counts.}
\Cref{tab:fi-msp-flip-appendix} uses \(R=120\) replicates per regime; \Cref{tab:directional-summary,tab:decision-rules} use \(R=200\) per \(\tau\); \Cref{tab:sca-vs-msp} uses \(R=80\) per \(\tau\). The differences reflect compute budgets: the 200-replicate runs omit joint SCA resampling, whereas the SCA comparison uses \(P=300\) null-forcing row resamples per replicate. All main simulations use \(n=800\), so independent runs can differ slightly at the same \(\tau\).

\paragraph{CI construction ablation on the LaLonde raw-earnings subset.}
Holding fixed the admissible \(K=4\) grid (\(\lvert S\rvert=12\)), the point-estimation pipeline, the seed, and a common bootstrap index matrix (\(B=500\)), percentile, bias-corrected, and bootstrap-Wald intervals give the same \(\mathrm{MSP}=1\), the same nearest configuration (full covariates), and the same three configurations whose CIs contain zero, so the interval construction moves the endpoints but not the conclusion.

\paragraph{False-positive rates for practical adequacy.}

The simulation of \Cref{sec:simulation} also has an operational consequence in this design, quantified in \Cref{tab:decision-rules}. We fix the practical-adequacy target \(\tau\ge0.7\), so the \(\tau=0.3\) arm contains a genuine positive effect that is below the adequacy threshold; it is not a no-effect null. Accordingly, ``false positive'' here means a false-positive classification for practical adequacy in this simulated design, not false detection of a nonzero effect, and not an inferential guarantee at the threshold \(\tau\ge0.7\): the radius test of \Cref{cor:radius-test} concerns directional support at zero, not practical adequacy. We declare a result robust at distance two when \(\mathrm{MSP}_{\mathrm{dir}}\ge2\), that is, when no configuration within one declared decision of the baseline has \(L(s)\le0\). The rule produces no such false-positive classifications at \(\tau=0.3\), no power at \(\tau=0.7\), and full power at \(\tau=1.5\). This is a stringent robustness certificate, not a general classifier. For transparent comparison, the table also displays the standard directionally significant-share summary at the diagnostic cutoff \(>0.5\). Its operating characteristics coincide with the rule's in this realized design; the cutoff was not fixed in advance as a decision rule, so these entries are an operating-characteristic comparison rather than a claim of dominance. The two non-directional configuration-count rules certify 86--90\% of the weak-effect replicates they should reject, yet only 2.5--4.0\% of the moderate-effect replicates they should accept.

\begin{table}[t]
    \centering
    \begin{tabular}{@{}lccc@{}}
        \toprule
        rule & FPR & TPR (.7) & TPR (1.5) \\
        \midrule
        \(\mathrm{MSP}_{\mathrm{dir}}\ge2\) & \textbf{0.000} & 0.000 & 1.000 \\
        \shortstack[l]{significant positive\\share \(>0.5\)} & \textbf{0.000} & 0.000 & 1.000 \\
        \shortstack[l]{significant\\share \(>0.9\)} & 0.860 & 0.025 & 0.105 \\
        \shortstack[l]{null-compatible\\share \(<0.2\)} & 0.895 & 0.040 & 0.465 \\
        \bottomrule
    \end{tabular}
    \caption{False-positive classification for practical adequacy (\(\tau\ge0.7\)) in this simulated design (\(n=800\), 200 replicates per \(\tau\)). The \(\mathrm{MSP}_{\mathrm{dir}}\ge2\) rule is the reported certificate; other thresholds are diagnostic comparators.}
    \label{tab:decision-rules}
\end{table}

Because \(\mathrm{MSP}_{\mathrm{dir}}=\min\{\mathrm{MSP}_{\mathrm{null}},\mathrm{MSP}_{\mathrm{rev}}\}\), the rule resolves the infinite branch directly: it rejects the two \(\tau=0.3\) replicates whose finite null-contact distance is at least two but whose reversal distance is one, and it certifies the 16 \(\tau=1.5\) replicates with pair \((\infty,\infty)\) that a blanket convention for \(\mathrm{MSP}_{\mathrm{null}}=\infty\) would reject. ROC rankings of scalar summaries are supplementary rather than the reporting rule; resolving \(\mathrm{MSP}=\infty\) by \(\mathrm{MSP}_{\mathrm{rev}}\) lifts AUC from 0.941 to 0.972.

\paragraph{SCA joint test.}
The SCA joint test~\citep{simonsohn2020specification} asks whether the full curve is inconsistent with a no-effect null. Treatment labels are not exchangeable in this simulation because assignment depends on covariates, so we do \emph{not} shuffle labels. Instead, we use the authors' non-experimental null-forcing procedure: for each specification \(k\), construct \(Y_k^*=Y-\hat\tau_kA\), jointly bootstrap rows of the resulting \(K\)-outcome data matrix, and refit all \(K\) specifications in each resample. We use \(P=300\) null resamples per replicate (\(80\) replicates per \(\tau\)). The reported directional statistics are the median estimate, the share significant in the predicted direction, and the aggregate Stouffer Z; their p-values are null-resampling p-values, not exact randomization-test p-values.

\begin{table}[h]
    \centering
    \begin{tabular}{@{}rccc@{}}
        \toprule
        \(\tau\) & median & sig.-share & Stouffer \\
        \midrule
        0.0 & 0.000 & 0.013 & 0.000 \\
        0.3 & 0.000 & 0.287 & 0.000 \\
        0.7 & 0.013 & 0.412 & 0.087 \\
        1.5 & \textbf{1.000} & \textbf{1.000} & \textbf{1.000} \\
        \bottomrule
    \end{tabular}
    \caption{Null-forcing SCA rejection rates (\(\alpha=0.05\), 80 replicates per \(\tau\), \(P=300\)). They assess global no-effect evidence, not distance to loss of support.}
    \label{tab:sca-vs-msp}
\end{table}

The two approaches answer different questions and neither table is a classifier comparison. At \(\tau=0.3\), the effect is nonzero but below the practical-adequacy target, so the 28.7\% significant-share rejection rate is evidence against a global no-effect null, not a false positive for adequacy. At \(\tau=0.7\), the three statistics have distinct power (1.3\%, 41.2\%, and 8.7\%). MSP is not a global SCA test: it adds the distance and the witness configuration, while \Cref{cor:radius-test} concerns a prespecified population radius rather than the SCA joint null. Reporting both therefore distinguishes global evidence from the geometry of loss of support without claiming that one summary universally dominates the other.

\paragraph{Bootstrap sensitivity check.}
MSP is computed from CIs, so its values depend on the number of bootstrap resamples \(B\). We verify stability by rerunning the comparison experiment with \(B=300\) (versus \(B=100\) in the main simulation of \Cref{tab:directional-summary}), using identical seeds, sample sizes, and replicate counts.

\begin{table}[h]
    \centering
    \begin{tabular}{@{}r cc cc@{}}
        \toprule
        & \multicolumn{2}{c}{\(P(\mathrm{MSP}=\infty)\)}
        & \multicolumn{2}{c}{median MSP (finite)} \\
        \cmidrule(lr){2-3}\cmidrule(lr){4-5}
        \(\tau\) & \(B=100\) & \(B=300\) & \(B=100\) & \(B=300\) \\
        \midrule
        0.0 & 0.085 & 0.060 & 0 & 0 \\
        0.3 & 0.815 & 0.810 & 1 & 1 \\
        0.7 & 0.025 & 0.010 & 1 & 1 \\
        1.5 & 0.080 & 0.090 & 2 & 2 \\
        \bottomrule
    \end{tabular}
    \caption{Bootstrap sensitivity of finite MSP and \(P(\mathrm{MSP}=\infty)\) at \(B=100\) and \(B=300\).}
    \label{tab:ci-sensitivity}
\end{table}

The qualitative pattern is identical across both settings (\Cref{tab:ci-sensitivity}). \(P(\mathrm{MSP}=\infty)\) differs by at most 0.025 in absolute terms. Median finite MSP is unchanged in all four regimes. Dispersion is unaffected by \(B\) by construction. These comparisons confirm that the main findings are not an artefact of bootstrap granularity. Apart from the CI-construction ablation and the simultaneous bands, every interval in the paper is a percentile bootstrap interval, held fixed across specifications to isolate the geometry of the specification space rather than differences in interval construction. We therefore do not claim CI-method invariance. Instead, \Cref{tab:ci-sensitivity} shows that the main patterns are stable to materially finer bootstrap granularity.

\paragraph{MSP--\(\alpha\) sensitivity.}
\label{par:msp-alpha}
Re-thresholding the stored bootstrap matrix of the raw-earnings LaLonde grid (\(K=4\), \(B=4{,}000\)) rather than re-bootstrapping, \(\mathrm{MSP}=1\) at every conventional level from \(\alpha=0.005\) to \(\alpha=0.10\), so the fragility reported in the main text is not an artefact of the 95\% convention.

\paragraph{Reproducibility, assets, and impacts.}
The exact declaration stress test is implemented in
\texttt{experiment/block13\_declaration\_stress\_test.py}. It reads the
stored 54-configuration results and the common \(4{,}000\)-draw bootstrap
matrix, enumerates all 256 product-form sub-declarations for each of four
baselines, and writes the pointwise, witness, worst-case, and simultaneous
diagnostics under \texttt{experiment/results\_block13\_declaration\_stress\_test/}.
The sup-\(t\) implementation is regression-tested against the main-text
12-configuration grid (\(q=3.16418\), envelope \([1,\infty]\)) before the
full stress test is run.
The anonymized repository linked in the submission contains the experiment scripts, pinned Python requirements, generated CSV outputs, and build instructions. The core entry point is \texttt{experiment/run\_paper\_experiments.py}; by default it uses paper-specific seeds, one worker, and fixed BLAS/OpenMP thread counts for deterministic execution within the pinned environment. The README separately gives the two-command enriched-LaLonde workflow: \texttt{block11\_multilevel\_axes.py --semantics B --B 4000} creates the fixed 54-configuration/bootstrap input, and \texttt{block13\_declaration\_stress\_test.py} consumes it for the exact declaration and weighted-cost diagnostics. All experiments were run on an Apple M2 (8-core CPU, 16\,GB unified memory) with no GPU acceleration. The core runner finishes in under one hour wall-clock time. The high-budget exceptions are the enriched \(B=4{,}000\) LaLonde run and the Block 12 noise-probe calibration, each of which takes several hours on the same hardware; the latter has checkpoint--resume support. Smoke-test modes use reduced replicate and resample counts for dependency checks.

The empirical illustrations use public benchmark data: LaLonde/NSW earnings data as distributed through standard public sources and the Card--Krueger minimum-wage data. We cite the original data sources in the paper and include scripts that document how each dataset is obtained or parsed. No new human-subjects data, scraped dataset, or model artifact is released.

\paragraph{Broader impacts.}
MSP is a methodological tool for empirical causal inference. Its primary societal benefit is improving the transparency and auditability of causal claims in policy-relevant research: by reporting the minimum number of coordinated analyst-decision changes needed to overturn a conclusion, researchers and readers can assess whether a result rests on genuine robustness or on an unexplored region of the specification space. This is especially relevant in high-stakes domains such as economics, medicine, and public policy, where the same dataset is often analyzed under different defensible choices and conclusions can be contested.

A potential negative use is selective reporting: an analyst could construct a narrow specification space designed to maximize MSP, creating the appearance of robustness without genuine insensitivity. To guard against this, MSP should be reported alongside a transparent, prespecified description of the axes and their justification. Another risk is over-interpreting large MSP values in observational settings as causal proof; the permutation references reported here are descriptive and do not provide causal guarantees. Responsible use requires acknowledging both the conditional nature of the metric and the limits of the declared specification space.

\end{document}